\documentclass[11pt]{article}

\usepackage{amsthm,amsmath,amssymb}
\RequirePackage[numbers]{natbib}
\RequirePackage[colorlinks,citecolor=blue,urlcolor=blue]{hyperref}
\usepackage{array}
\usepackage{algorithm}
\usepackage{algorithmic}
\usepackage{color}
\usepackage{times}
\usepackage{graphicx}
\usepackage{natbib}

\newtheorem{theorem}{Theorem}
\newtheorem{corollary}{Corollary}
\newtheorem{lemma}{Lemma}

\newtheorem{assumption}{Assumption}

\def\fac{F_q( \xi)}
\def\bfac{\bar{F}_q( \xi)}
\def\cone{\mathcal{C} ( \xi)}
\def\Pr{\mathbb P}
\def\Ex{\mathbb E}

\def\mA{\mathcal{A}}
\def\sign{{\rm sign}}
\def\indyk{\mathbb{I}}
\def\hbeta{\tilde{\theta}}
\def\th*{\theta^*}
\def\hth{\hat{\theta}}
\def\ahth{\hat{\theta}^a}
\def\tth{\tilde{\theta}}
\def\thmin{\theta^0_{min}}
\def\t0{\theta^0}
\def\thth{\hat{\theta}^{th}}

 \title{Rank-based Lasso - efficient methods for high-dimensional robust model selection}

\author{
Wojciech Rejchel 
\footnote{Nicolaus Copernicus University,  Toru\'n, Poland, 
wrejchel@gmail.com. } 
Ma\l gorzata Bogdan \footnote{University of Wroc{\l}aw, Wroc\l{}aw, Poland and
Lund University, Lund, Sweden, Malgorzata.Bogdan@uwr.edu.pl. The research was funded by the NCN grant 2016/23/B/ST1/00454.}
}

\date{}

\begin{document}

\maketitle

\begin{abstract}
 We consider the problem of identifying significant predictors in large data bases, where the response variable depends on the linear combination of  explanatory variables through an unknown link function, corrupted with the noise from the unknown distribution. We utilize the natural, robust and efficient approach, which relies on replacing values of the response variables by their ranks and then identifying significant predictors by using  well known Lasso. We provide new consistency results for the proposed procedure (called ,,RankLasso'') and extend the scope of its applications by proposing its thresholded and adaptive versions.  Our theoretical results show that these modifications can identify the set of relevant predictors under  much wider range of data generating scenarios than regular RankLasso. 
Theoretical  results are supported by the simulation study and the real data analysis, which show that our methods can properly identify relevant predictors, even when the error terms come from the Cauchy distribution and  the link function is nonlinear. They also demonstrate the superiority of the modified versions of RankLasso in the case when predictors are substantially correlated.   The numerical study shows also that  RankLasso performs substantially better in model selection than LADLasso, which is a well established methodology for robust model selection.
\end{abstract}

List of keywords: Lasso, model selection, ranks, single index model, sparsity, $U$-statistics

\section{Introduction}
\label{intro}

Model selection is a fundamental challenge when working with large-scale data sets, where 
the number of predictors  exceeds significantly the number of observations. In many practical problems finding a small set of significant predictors is at least as important as  accurate estimation or prediction. Among many approaches to high-dimensional model selection one can distinguish a large group of methods  based on penalized  estimation \citep{els:01, geer:11}. Under the multiple linear regression model
$$
Y_i= \beta ' X_i + \varepsilon _i, \quad i=1,\ldots,n,
$$
where $Y_i \in \mathbb{R}$ is a response variable, $X_i \in \mathbb{R}^p$ is a vector of predictors,
$ \beta\in R^p$ is the vector of model parameters and $\varepsilon_i$ is a random error,  
the penalized model selection approaches usually recommend estimating the vector of regression coefficients $\beta$ by 
\begin{equation}
\label{Pen}
 \hat \beta =\arg   
\min_{\beta \in \mathbb{R}^p}  \quad 
\sum_{i=1}^n (Y_i -  \beta ' X_i)^2 + Pen(\beta),
\end{equation}
where $\sum_{i=1}^n (Y_i -  \beta ' X_i)^2$ is the quadratic loss function measuring the model fit  and $Pen(\beta)$ is the penalty on the model complexity. 
 The main representative of these methods is  Lasso \cite{tib:96}, which uses the $l_1$-norm penalty. The properties of Lasso in model selection, estimation and prediction are deeply investigated, e.g. in \cite{mbuhl:06, zhaoyu:06, zou:06,geer:08, bickel:09, yezhang:10,  geer:11, 
huang:13, TardivelBogdan2018}. These articles discuss the properties of Lasso in the context of linear or generalized linear models and their results hold under specific assumptions on the relationship between the response and explanatory variables and/or the distribution of the random errors. However, it is quite common that a complex data set does not satisfy these assumptions or they are difficult to verify.  In such cases it is advised to use ,,robust'' methods of  model selection.

In this paper we consider the single index model  
\begin{equation}
\label{model}
Y_i=g( \beta ' X_i , \varepsilon _i), \quad i=1,\ldots,n,
\end{equation}
where $g$ is unknown link function. Thus, we suppose that predictors influence the response variable through the link function $g$ of the scalar product $\beta ' X_i$. However, we make no assumptions on the form of the link function $g$ nor on the distribution of the error term $\varepsilon _i$. Specifically, we do not assume the existence of the expected value of $\varepsilon_i$.

The goal of the model selection is the identification of the set of relevant predictors
\begin{equation}
\label{Tbeta}
T=\{1 \leq j \leq p: \beta_j\neq 0\}.
\end{equation}
The literature on the topic of robust model selection is quite considerable and the comprehensive review can be found e.g. in \cite{WuMa15}. 
Many of the existing methods suppose that the linear model assumption is  satisfied and consider the robustness with respect to the noise. Here the most popular approaches rely on replacing the regular quadratic  loss function with the loss function which is more robust with respect to  outliers, like e.g. the absolute value or Huber \cite{Huber} loss functions.
 Model selection properties of the penalized regression procedures with such robust loss functions were investigated, among others, in \cite{WangLiJiang2007, GaoHuang10, belloni11,WangWuLi12, Wang13, fanfanbarut14, PengWang15, Zhongetal16, MedinaRon18}. Among these methods one can mention the approach of \cite{JohnsonPeng08, Johnson09}, where the loss function is expressed in terms of residual's ranks.
On the other hand, the issues of model selection in misspecified models were discussed in e.g.  
 \cite{LuGoldFine12, LvLiu14}, while   robustness with respect to the unknown link function $g$ in the single index model \eqref{model} was discussed, for instance, in \cite{KongXia07, ZengHeZhu12, Alquier13, ChengZengZhu2017}. 
Model selection procedures that are robust with respect to both the link function and the distribution of the noise were developed e.g. in \cite{ZhuZhu09,songma:10,WangZhu2015, Zhongetal16, Rejchel:17, Rejchel17Neuro}. 
 The results obtained in these and many other papers confirm good properties of the penalized robust model selection procedures. However, the application of  procedures based on robust loss functions (e.g. piecewise-linear)  in the context of the analysis of large data sets is often limited due to their computational complexity and/or the need of the development of dedicated optimization algorithms. For instance, in Sections \ref{numerical} and \ref{real} we consider the Least Absolute Deviation Lasso (LADLasso) estimator \citep{WangLiJiang2007, belloni11, fanfanbarut14}, which turns out to be computationally very slow even for moderate dimension experiments. 

In the current paper we consider an alternative approach that is computationally fast and can work efficiently with complex high-dimensional data sets. 
Our procedure is very simple and relies on replacing actual values of the response variables $Y_i$
by their centred ranks. Ranks $R_i$ are defined as 
\begin{equation}
\label{rank}
R_i = \sum_{j=1}^n \indyk (Y_j \leq Y_i), \quad i=1,\ldots,n\;\;,
\end{equation}
where $\indyk (\cdot)$ is the indicator function.  Next, we identify significant predictors by simply solving the following Lasso problem;
\begin{equation}
\label{Rlasso}
\mbox{\bf{RankLasso:}} \quad \quad \quad \hat \theta =\arg   
\min_{\theta \in \mathbb{R}^p}  \quad 
Q(\theta)+\lambda \left|\theta \right| _{1},
\end{equation}
where 
\begin{equation}
\label{Q_rand}
Q(\theta) = \frac{1}{2n} \sum_{i=1}^n \left(R_i/n - 0.5  - \theta ' X_i \right)^2.
\end{equation}
This procedure does not require any dedicated algorithm and can be executed using efficient implementations of Lasso in ,,R'' \citep{R:11} packages ,,lars'' \citep{lars:04} or ,,glmnet'' \citep{glmnet:10}.
    
Replacing values of  response variables by their ranks is a well-known approach in non-parametric statistics and leads to robust procedures. The premier examples of  a rank approach are the Wilcoxon test, that is a widely used alternative to the Student's t-test, or the Kruskall-Wallis ANOVA test.
The similar high-dimensional rank model selection procedure with the $l_0$-penalty was proposed in  \cite{ZakBogdan07} in the context of the analysis of genetic data and its good properties were confirmed in the simulation study reported in \cite{mBICranks}. 

 In Subsection \ref{unknown_link} we provide the definition of the parameter $\theta^0$, that is estimated by RankLasso, and discuss its relationship to $\beta$. It turns out that under certain standard assumptions, the support of $\theta^0$ coincides with the support of $\beta$ and RankLasso can identify the set of relevant predictors.
The fact that RankLasso can be used as an efficient tool for robust model selection was already observed in
\cite{ZhuZhu09,WangZhu2015}, who proposed this procedure using a slightly more complicated definition, which makes it difficult to notice the relationship with ranks.
In  \cite{WangZhu2015}   model selection consistency of  RankLasso is proved under (so called) ,,irrepresentable conditions'', that are  required for  model selection consistency of regular Lasso  in generalized linear models \cite{zhaoyu:06}. These assumptions are very restrictive (see e.g. \citep{zhaoyu:06, conditions:2009}) and are hardly satisfied in practice.  On the other hand, it is well-known that  regular Lasso is able to screen predictors or even properly separate them  under much less restrictive assumptions \citep{yezhang:10, geer:11, TardivelBogdan2018}.

In this paper we refine the theoretical results of \cite{ZhuZhu09, WangZhu2015} and extend the scope of applications of RankLasso by proposing its thresholded and adaptive versions. In the case of   standard Lasso  similar modifications were introduced and discussed e.g. in \cite{zou:06, Candess2008, zhou:09}. We prove that these modifications are model selection consistent in the model \eqref{model} under much weaker assumptions than RankLasso. These results  require a substantial modification of the proof techniques as compared to the similar results for  regular Lasso. It is related to the fact that
ranks are dependent, so  \eqref{Q_rand} is a sum of dependent random variables. In Subsection \ref{Ustat_sec} we describe how this problem can be overcome with the application of the  theory of $U$-statistics. We also present extensive numerical results  illustrating that the modifications of  RankLasso can indeed properly identify relevant predictors, even  if the link function is not linear, the error terms come from, say, the Cauchy distribution and predictors are substantially correlated.  These results also show  that  RankLasso  compares favorably with LADLasso, which is a well established methodology for robust model selection \citep{WangLiJiang2007, belloni11, fanfanbarut14}. 

The paper is organized as follows: in Section \ref{asymp_prop} we present theoretical results on the model selection consistency of RankLasso and its modifications. In Subsection \ref{unknown_link} we discuss the relationship between $\beta$ and the parameter estimated by RankLasso. We show that our approach is able to identify the support of $\beta$ despite the fact that this parameter is not identifiable in the single index model. In Subsection \ref{high_sec}  we consider properties of estimators in the high-dimensional scenario, where the number of predictors can be much larger than the sample size. We establish nonasymptotic bounds on the estimation error and separability of RankLasso. We  use these results to prove model selection consistency of  thresholded and weighted RankLasso. In Subsection \ref{Ustat_sec} we briefly draw a road map to the proofs of main results. Sections 
\ref{numerical} and \ref{real} are devoted to experiments  that illustrate the properties of  rank-based estimators on simulated and real data sets, respectively. The proofs of main and auxiliary results are relegated to the appendix. Due to the space limit, we also place in the appendix results for  the low-dimensional case, where the number of predictors is fixed and the sample size diverges to infinity. 
In this case we provide the necessary and sufficient conditions for model selection consistency of RankLasso and much weaker sufficient conditions for this property  of thresholded and weighted versions of RankLasso.

\section{Model selection properties of RankLasso and its modifications}
\label{asymp_prop}

In this section we provide theoretical results concerning  model selection and estimation properties of  RankLasso and its thresholded and weighted versions. We start with specifying the assumptions on our model.

\subsection{Assumptions and notation}
 

Consider the single index model (\ref{model}). 
In this paper we assume that the design matrix $X$ and the vector of the error terms $\varepsilon$ satisfy  the following assumptions.
\begin{assumption}\label{as1}
We assume that $(X_1,\varepsilon_1), \ldots, (X_n,\varepsilon_n)$ are i.i.d. random vectors such that the distribution of $X_1$ is absolutely continuous and $X_1$ is independent of the noise variable $\varepsilon_1$. Additionally, we assume that $\Ex X_1=0$, $H=\Ex X_1X_1 '$ is positive definite and $H_{jj}=1$ for $j=1,\ldots,p.$ 
\end{assumption}

 The single index model (\ref{model}) does not allow to estimate an intercept and can identify  $\beta$ only up to a multiplicative constant, because  any shift or scale change in $\beta ' X_i$ can be absorbed by $g$. However, in many situations RankLasso can properly identify the support $T$  of $\beta$. In this paper we will prove this fact under the following assumption. 

\begin{assumption}\label{as2} We assume that for each $\theta \in \mathbb{R}^p$ the conditional expectation $\Ex (\theta ' X_1 | \beta ' X_1)$ exists and  
$$
\Ex (\theta ' X_1 | \beta ' X_1) = d_ \theta \beta ' X_1
$$
for a real number $d_\theta \in \mathbb{R}.$ 
\end{assumption}

Assumption \ref{as2} is a standard condition in the literature on the single index model or on the model misspecification (see e.g. \citep{Brill:83, Ruud83, LiDuan:89, ZhuZhu09, WangZhu2015, Zhongetal16, KubMiel17}). It is always satisfied in the simple regression models (i.e. when $X_1 \in \mathbb{R}$), which are often used for initial screening of explanatory variables (see e.g.\citep{Fan2008}). It is also satisfied whenever $X_1$ comes from the {\it elliptical distribution}, like the
multivariate normal distribution or multivariate $t$-distribution. 
The interesting paper \cite{HallLi93} advocates that the Assumption \ref{as2} is a nonrestrictive condition when the number of predictors is large, which is the case that we focus on in the paper. In the experimental section of this article we show that RankLasso, proposed here, is able to identify the support of $\beta$ also when the columns of the design matrix contain genotypes of independent Single Nucleotide Polymorphisms, whose distribution is not symmetric and clearly does not belong to the elliptical distribution.

The identifiability of the support of $\beta$ by the rank procedure requires also the assumptions on the monotonicity of the link function $g$ and the cumulative distribution function of $Y_1$. The following Assumption \ref{as_main}, which combines Assumptions \ref{as1} and \ref{as2} and the monotonicity assumptions, will be used in most of theoretical results in this article.

\begin{assumption}\label{as_main}
We assume that the design matrix and the error term satisfy  Assumptions \ref{as1} and \ref{as2}, the cumulative distribution function $F$ of the response variable $Y_1$ is increasing and $g$ in \eqref{model} is  increasing   with respect to the first argument.
\end{assumption}

In this paper we will use the following notation: \\
- $X=(X_1, X_2, \ldots, X_n)'$ is the $(n \times p)$-matrix of predictors, \\
- $\bar{X} = \frac{1}{n} \sum_{i=1} ^n X_i$,\\
- $Z_i = (X_i, Y_i), i=1,\ldots,n$, \\
- $T'=\{1,\ldots,  p  \} \setminus T$ is a complement of $T$,\\
- $X_T$ is  submatrix of $X,$ with columns whose indices belong to  support $T$ of $\beta$ (see \eqref{Tbeta}),\\
- $\theta_T$ is a restriction of a vector $\theta \in \mathbb{R}^p$ to the indices from $T,$\\
- $p_0$ is the number of elements in $T,$\\
- the $l_q$-norm of a vector is defined as $|\theta|_q= \left(\sum_{j=1}^p|\theta _j|^q\right)^{1/q}$ for $q \in [1,\infty].$

\subsection{Identifying the support of $\beta$}
\label{unknown_link}

RankLasso does not estimate $\beta$ but the vector  
\begin{equation}
\label{mint0}
\t0 = \arg \min_{\theta \in \mathbb{R}^p} \quad \Ex \, Q(\theta),
\end{equation}
where $Q(\theta)$ is defined in (\ref{Q_rand}).
Since $H$ is positive definite, the minimizer $\t0$ is unique and is given by the formula
\begin{equation}
\label{th_ranks}
\t0= \frac{1}{n^2} H^{-1} \left(\Ex \sum_{i=1}^n R_i X_i \right).
\end{equation}
Now, using the facts that 
\begin{equation}
\label{sum_rank}
\sum_{i=1}^n R_i X_i = \sum_{i=1}^n \sum_{j=1}^n \indyk (Y_j \leq Y_i)X_i = \sum_{i \neq j} \indyk (Y_j \leq Y_i)X_i  + \sum_{i=1}^n X_i
\end{equation}
and that $EX_i=0$,  we can write 
\begin{equation}
\label{theta0_alt}
\t0 = \frac{n-1}{n} H^{-1}\mu,
\end{equation}
where $\mu= \Ex \left[\indyk (Y_2 \leq Y_1)X_1\right]$ is the expected value of the $U$-statistic
 \begin{equation}
\label{An}
A= \frac{1}{n(n-1)}\sum_{i \neq j} \indyk (Y_j \leq Y_i)X_i\;\;.
\end{equation} 
In the next theorem we state the relation between $\theta^0$ and $\beta.$

\begin{theorem}
\label{mult_cor}
Consider the model \eqref{model}.
If Assumptions \ref{as1} and \ref{as2} are satisfied, then  $$\t0=\gamma_\beta \beta$$
with 
\begin{equation}
\label{gamma}
\gamma_\beta =  
\frac{\frac{n-1}{n}  \beta ' \mu}{\beta ' H \beta}=
\frac{\frac{n-1}{n} \, Cov(F(Y_1) ,\beta ' X_1)}{\beta ' H \beta}\:,
\end{equation} 
where $F$ is a cumulative distribution function of a response variable $Y_1.$ 

Additionally, if $F$ is increasing and $g$ is increasing with respect to the first argument, then 
$\gamma_\beta>0$, so the signs of $\beta$ coincide with the signs of $\t0$ and 
\begin{equation}
\label{Ttheta}
T=\{j: \beta _j \neq 0\} = \{j: \t0 _j \neq 0\}.
\end{equation}
\end{theorem}

We can apply Theorem \ref{mult_cor} to the additive model 
$
Y_i = g_1(\beta ' X_i) + \varepsilon_i 
$
with an increasing function $g_1$. Then, under Assumptions \ref{as1} and \ref{as2},  $\t0=\gamma_\beta \beta$ with $\gamma_ \beta  >0$ if, for example,
 the support of the noise variable is a real line. 
Moreover, since the procedure based on ranks is invariant with respect to increasing transformations of response variables the same applies to the model
$Y_i = g_2(\beta ' X_i + \varepsilon_i)$ with an increasing function $g_2$.

\subsection{High-dimensional scenario} 
\label{high_sec}

In this subsection we consider properties of the RankLasso estimator and its modifications in the case where the number of predictors can be much larger than the sample size. 
To obtain the results of this subsection we need the additional condition:

\begin{assumption}\label{as_sub}
 We suppose that the vector of significant predictors $(X_1)_T$ is subgaussian with the coefficient $\tau _0>0$, i.e. for each $u\in \mathbb{R}^{p_0}$ we have $\Ex \exp(u' (X_1)_T) \leq \exp(\tau_0^2 u'u/2).$ Moreover, the irrelevant predictors are univariate subgaussian, i.e.  
for each $a \in \mathbb{R} $ and $j \notin T$ we have $\Ex \exp(a X_{1j}) \leq \exp(\tau_j^2 a^2/2)$ for positive numbers $\tau _j.$ Finally, we denote $\tau = \max (\tau_0, \tau_j, j \notin T).$
\end{assumption}

We need subgaussianity of the vector of predictors to obtain exponential inequalities in the proofs of the
main results in this subsection. This condition is a standard assumption while working with random predictors \citep{Raskutti:2010, Cox13, BuhlmannGeer2015} in high-dimensional models.

\subsubsection{Estimation error and separability of RankLasso}
\label{unweight_high}

 Model selection consistency of RankLasso in the high-dimensional case was proved  in \cite[Theorem 2.1]{WangZhu2015}. However, this result requires the stringent  irrepresentable condition. Moreover, it is obtained under the polynomial upper bound on the dependency of $p$ on $n$ and provides only a rough guidance of selection of the tuning parameter $\lambda$.  
In our article we concentrate on estimation consistency of RankLasso, which paves the way for  model selection consistency of the weighted and thresholded versions of this method. 
Compared to the asymptotic results of \cite{WangZhu2015} our results are stated in the form of non-asymptotic inequalities, they do not require the irrepresentable condition, allow for the exponential increase of $p$ as a function of $n$ and provide a  precise guidance on selection of  regularization parameter~$\lambda$.


We start with introducing the cone invertibility factor (CIF), that plays an important role in investigating properties of estimators based on the Lasso penalty \citep{yezhang:10}. 
In the case $n>p$ one usually uses the minimal eigenvalue of the matrix $X 'X/n$ to express the strength of correlations between predictors. Obviously, in the high-dimensional scenario this value is equal to zero and the minimal eigenvalue needs to be replaced by some other measure of predictors interdependency, which would describe the potential of consistent estimation of model parameters.

Let  $T$ be the set of indices corresponding to the support of the true vector $\beta$ and let $\theta_T$ and $\theta_{T'}$ be the restrictions of the vector $\theta \in \mathbb{R}^p$ to the indices from $T$ and $T',$ respectively. Now, for $\xi>1$ we consider a cone 
$$\cone = \{ \theta \in \mathbb{R}^p : |\theta_{T'} |_1 \leq \xi |\theta_{T} |_1\}\;\;.
$$

In the case when $p>>n$ three different characteristics measuring the potential for consistent estimation of the model parameters have been introduced:\\
- the restricted eigenvalue \citep{bickel:09}:\\
$$
RE(\xi) = \inf_{0 \neq \theta \in \cone } \frac{\theta' X'X\theta/n}{|\theta|^2_2}\;\;,
$$
- the compatibility factor \citep{geer:08}:
$$
K(\xi)=\inf_{0 \neq \theta \in \cone } \frac{p_0\theta' X'X\theta/n}{|\theta_T|_1^2}\;\;,
$$
- the cone invertibility factor (CIF, \citep{yezhang:10}): for $q\geq 1$
$$
\bfac = \inf_{0 \neq \theta \in \cone } \frac{p_0^{1/q}|X'X\theta/n|_\infty}{|\theta|_q}\;\;.
$$

In this article we will use CIF, since this factor allows for a sharp formulation of convergency results for all  $l_q$ norms with $q\geq 1,$ 
see \cite[Section 3.2]{yezhang:10}. 
The population (non-random) version of CIF is given by
$$
\fac = \inf_{0 \neq \theta \in \cone } \frac{p_0^{1/q}|H\theta|_\infty}{|\theta|_q}\:,
$$ 
where $H=E X_1X_1'$.
The key property of  the random and the population versions of CIF, $\bfac$ and $\fac$, is that,  in contrast to the smallest eigenvalues of matrices $X'X/n$ and $H,$  they can be close to each other in the high-dimensional setting, see \cite[Lemma 4.1]{Cox13} or \cite[Corollary 10.1]{conditions:2009}. This fact is used in the proof of Theorem~\ref{high} (given below).  

In the simulation study in Section \ref{numerical} we consider predictors which are independent  or equi-correlated, i.e.  $H_{jj}=1$ and  $H_{jk}=b$ for $j \neq k $ and  $b\in [0,1).$ In this case the smallest eigenvalue of $H$ is $1-b.$ For  $\xi > 1$ and $q \geq 2$ we can lower bound CIF by 
$$ F_q(\xi)  \geq (1+\xi)^{-1} p_0 ^{1/q-1/2}  (1-b)\;\;,$$ 
which illustrates that 	CIF diminishes with the increase of $p_0$ and $b$.  Also, in the case when $H=I$ and $q=\infty$, $F_q(\xi)=1$, independently of $\xi$ or $p_0$.

The next result describes the estimation accuracy of RankLasso.

\begin{theorem}
\label{high}
Let $a \in (0,1), q \geq 1$ and $\xi >1$ be arbitrary. Suppose that Assumptions \ref{as_main} and \ref{as_sub} are satisfied. 
Moreover, suppose that 
\begin{equation}
\label{n_cond}
n \geq \frac{K_1 p_0^2 \tau^4 (1+\xi)^2 \log(p/a)}{F_q^2(\xi)}
\end{equation}
and
\begin{equation}
\label{lambda}
\lambda \geq K_2 \frac{\xi+1}{\xi-1} \tau ^2 \sqrt{ \frac{ \log(p/a)}{\kappa n}} \:,
\end{equation}
where $K_1,K_2$ are universal constants and $\kappa$ is the smallest eigenvalue of the correlation matrix between the true predictors $H_T= \left(H_{jk}\right)_{j,k \in T}$.   Then there exists a universal constant $K_3 >0$ such that with probability at least $1- K_3 a$ we have 
\begin{equation}
\label{high_estim}
|\hth - \theta^0|_q \leq \frac{4 \xi p_0^{1/q} \lambda}{ (\xi +1) \fac} \:.
\end{equation}
Besides, if $X_1$ has a normal distribution $N(0,H),$ then $\kappa$ and $\tau$ can be dropped in \eqref{n_cond} and \eqref{lambda}.
\end{theorem}

In Theorem \ref{high} we provide the bound for the estimation error of  RankLasso.
This result also provides the conditions for the estimation consistency of RankLasso, which can be obtained by replacing $a$ with a sequence $a_n$ that decreases not too fast  and selecting the minimal $\lambda=\lambda_n$ satisfying the condition (\ref{lambda}). 
The  consistency holds in the high-dimensional scenario, i.e. the number of predictors can be significantly greater than the sample size. Indeed, the consistency in the $l_{\infty}$-norm holds e.g. when $p= \exp(n^{\alpha_1}), p_0=n^{\alpha_2}, a=\exp(-n^{\alpha_1}),$ where $\alpha_1 + 2 \alpha_2 <1 ,$ and $\lambda$ is equal to the right-hand side of the inequality \eqref{lambda}, provided that $F_{\infty} (\xi)$ and $\kappa$ are lower bounded (or slowly converging to 0) and $\tau $ is upper bounded (or slowly diverging to $\infty$).

Theorem \ref{high} is an analog of \cite[Theorem 3]{yezhang:10} that refers to the linear model with noise variables having a finite variance. Similar results for quantile regression, among others for LADLasso, can be found in \cite[Theorem 2]{belloni11} and \cite[Theorem 1]{Wang13}. 
Obviously, they do not need existence of the noise variance, but put some restrictions on the conditional distribution of $Y_1$ given $X_1$ (cf. \cite[Assumption D.1]{belloni11} or \cite[the condition (13)]{Wang13}). In Theorem \ref{high} we do not require such assumptions 
and we work in the single index model \eqref{model}, that the relation between $Y_1$ and predictors can be nonlinear.

The key property of the procedures based on the Lasso penalty is that they are able to screen predictors (the set $T$ is contained in the support of the Lasso estimator) or even  to separate relevant  and irrelevant predictors, if the signal is strong enough. In the next result we prove that  RankLasso also has this property. 
 
\begin{corollary}
\label{seperation}
Suppose that conditions of Theorem  \ref{high} are satisfied for $q= \infty.$
Let $\t0 _{min} = \min\limits_{j \in T} |\t0 _j|.$ 
If $ \t0 _{min} \geq  \frac{8 \xi  \lambda}{ (\xi +1) F_\infty (\xi)},$ then 
\[
P\left( \forall_{j \in T, k \notin T} \quad |\hth _j| > |\hth _k|\right) \geq 1- K_3a\,.
\]
\end{corollary}

The separation of predictors by RankLasso given in Corollary \ref{seperation} is a very useful property. It will be used to prove model selection consistency of the thresholded
and weighted RankLasso in the next part of the paper.

 Finally, we discuss the condition of Corollary \ref{seperation} that $\t0 _{min}$ cannot be to small, i.e. $ \t0 _{min} \geq  \frac{8 \xi  \lambda}{ (\xi +1) F_\infty (\xi)}$. Using Theorem \ref{mult_cor} we know that $\t0= \gamma_\beta \beta$ and $\gamma_\beta>0,$ so this condition refers to the strength of the true parameter $\beta$ and requires that 
\begin{equation}
\label{bamin}
\min_{j \in T} |\beta_j| \geq  \frac{ 8 \xi \lambda}{ \gamma_\beta (\xi +1) F_\infty (\xi)} \:.
\end{equation} 
 Compared to the similar condition for  regular Lasso in the multiple linear model, the denominator contains an additional factor $\gamma_\beta$. This number  is usually smaller than one,  so  RankLasso  needs  larger sample size (or stronger  signal) to work well.  This phenomenon is typical for the single-index model, where the similar restrictions hold  for competitive methods like e.g. LADLasso. 

Below we provide a simplified version of Theorem \ref{high}, formulated under the assumption that $F_{\infty} (\xi)$ and $\kappa$ are lower bounded and $\tau $ is upper bounded.  This formulation will be used in the following subsection to increase the transparency of the results on the model selection consistency of weighted and thresholded RankLasso. 

\begin{corollary}
\label{cor_simply}
Let $a \in (0,1)$ be arbitrary. Suppose that Assumptions \ref{as_main} and \ref{as_sub} are satisfied. Moreover, assume that  there exist $\xi_0>1$ and constants $C_1>0$ and $C_2<\infty$ such that $\kappa \geq C_1$, $F_\infty (\xi_0) \geq C_1$ and $\tau \leq C_2$. 
 If
\begin{equation*}
n \geq K_1 p_0^2  \log(p/a)
\end{equation*}
and 
\begin{equation*}
\lambda \geq K_2 \sqrt{ \frac{ \log(p/a)}{n}} \:,
\end{equation*}
then  
\begin{equation}
\label{bound}
P\left(|\hth - \theta^0|_\infty \leq 4  \lambda/C_1 \right)\geq 1- K_3 a\;\;,
\end{equation}
where the constants $K_1$ and $K_2$ depend only on $\xi_0, C_1, C_2$ and $K_3$ is a universal constant provided in Theorem \ref{high}.   
\end{corollary}

\subsubsection{Modifications of RankLasso}
\label{high_adap_lasso}

The main drawback of RankLasso considered in Subsection \ref{unweight_high} is that  it can recover the true model only if the restrictive irrepresentable condition is satisfied. If this condition does not hold, then RankLasso can achieve a high power only by including a large number of irrelevant predictors. In Theorems \ref{th_high} and  \ref{adap_high} we state that this problem can be overcome by the application of the weighted or thresholded versions of  RankLasso. In both cases we rely on the initial  RankLasso estimator $\hth$  of $\t0$, which is estimation consistent under the assumptions of Theorem \ref{high} or Corollary \ref{cor_simply}. 
Theorems \ref{th_high} and \ref{adap_high} are stated  under simplified assumptions of Corollary  \ref{cor_simply}. We have decided to establish them in these versions to make this subsection more communicable.

First, we consider thresholded RankLasso, which is denoted by
$\thth$ and defined as
\begin{equation}
\label{thresh}
\thth _j = \hth _j \indyk(|\hth _j| \geq \delta), \quad j=1, \ldots,p,
\end{equation}
where $\hth$ is the RankLasso estimator  given in (\ref{Rlasso}) and $\delta >0$ is a threshold. 
 Theorem~\ref{th_high} provides the conditions under which this procedure is model selection consistent.

\begin{theorem}
\label{th_high}
We assume that Corollary \ref{cor_simply} holds and that the sample size and the tuning parameter $\lambda$ for RankLasso are selected according to  Corollary~\ref{cor_simply}. Moreover, suppose that 
$\t0_{min} = \min\limits_{j \in T} |\t0 _j| $ is such that it is possible to select the threshold $\delta$ so as
$$\t0 _{min}/2 \geq \delta > K_4 \lambda,$$ where $K_4=4/C_1$ is the constant from \eqref{bound}.  Then it holds  
$$
P\left( \hat{T}^{th} = T \right) \geq 1- K_3 a,$$
where  $K_3$ is the universal constant from Theorem \ref{high} and $\hat{T}^{th}= \{1\leq j \leq p: \thth _j \neq 0\}$ is the estimated set of relevant predictors by  thresholded RankLasso. 
  
\end{theorem}

Next, we consider the weighted RankLasso that minimizes
\begin{equation}
\label{gamma_a}
Q(\theta) + \lambda_a \sum_{j=1}^p w_j |\theta_j| ,
\end{equation}
where $\lambda_a>0$ and  weights are chosen in the following way:
 for arbitrary number $K>0$ and the RankLasso estimator $\hth$ from the previous subsection    we have $w_j=|\hth _j|^{-1}$ for  $|\hth _j| \leq \lambda_a,$ and $w_j \leq K,$ otherwise.

The next result describes properties of the weighted RankLasso estimator.

\begin{theorem}
\label{adap_high}
We assume that Corollary \ref{cor_simply} holds and that the sample size and the tuning parameter $\lambda$ for RankLasso are selected according to  Corollary \ref{cor_simply}. 
Let
 $\lambda_a = K_4 \lambda,$ where $K_4=4/C_1$ is from  \eqref{bound}. Additionally, we suppose that the signal strength and sparsity satisfy $\thmin /2 > \lambda _a$ and $p_0 \lambda \leq  K_5,$ where $K_5 $ is sufficiently small constant.
Then with probability at least $1- K_6 a$ there exists a global minimizer $\hth ^a$ of \eqref{gamma_a} such that $\hth ^a _{T'} =0$ and
\begin{equation}
\label{adap_formula}
|\hth ^a _T - \theta^0_T|_1 \leq  K_7 p_0 \lambda \:,
\end{equation}
where $K_6$ and $K_7$ are the constants depending only on $K_1,\ldots,K_5$ and the constant $K$, that is used in the definition of weights.
\end{theorem}

In \cite[Corollary 1]{fanfanbarut14} the authors considered the weighted Lasso with the absolute value loss function in the linear model. Thus, this procedure is  robust with repect to the distribution of the noise variable. However, working with the absolute value loss function they need  that, basically, the density of the noise is Lipschitz in a neighbourhood of zero \citep[Condition 1]{fanfanbarut14}.  Our thresholded and weighted RankLasso  do not require such restrictions.
Besides, Theorems \ref{th_high} and \ref{adap_high} confirm that the proposed procedures works well in model selection in the single index model \eqref{model}.


It can be seen in the proof of Theorem \ref{adap_high} that $K_7$ is an increasing function of $K,$ that occurs in the construction of weights. It is intuitively clear, because weights $w_j\leq K$ usually correspond
(by Corollary \ref{seperation}) to significant predictors. Therefore, increasing $K$  we shrink coordinates of the estimator, so the bias increases. This fact is described in  \eqref{adap_formula}.

\subsection{Road map to proofs of main results}
\label{Ustat_sec}

In the paper we study properties of the RankLasso estimator and its thresholded and weighted modifications. These estimators are obtained by minimization of the risk $Q(\theta)$ defined in \eqref{Q_rand} and the penalty. 
Therefore, the analysis of model selection properties of rank-based estimators is based on investigating these two terms. The penalty term can be handled using standard methods for regular Lasso. However, the analysis of the risk $Q(\theta)$ is different, because it is
 the sum of dependent variables. 
In this subsection we show that the theory of $U$-statistics \citep{hoeff:48, serf:80, pg:99}
plays a prominent role in studying properties of the risk $Q(\theta)$. 

Consider the key object that is the derivative of the risk at the true point $\t0$
\begin{equation}
\label{derivQ}
\nabla Q(\t0) = - \frac{1}{n^2} \sum_{i=1}^n R_i X_i + \frac{1}{2n} \sum_{i=1}^n X_i + \frac{1}{n} \sum_{i=1}^n  X_i X_i' \t0 .
\end{equation}
The second and third terms in \eqref{derivQ} are  sums of independent random variables, but the first one is a sum of dependent random variables. 
Using \eqref{sum_rank}  we can express \eqref{derivQ} as
\begin{equation}
\label{derivQu}
-\frac{n-1}{n} A + \frac{n-2}{2n^2}  \sum_{i=1}^n X_i  + \frac{1}{n} \sum_{i=1}^n  X_i X_i' \t0
\end{equation}
where  a $U$-statistic $A$ is defined in \eqref{An} and has the kernel 
\begin{equation}
\label{kernel}
f(z_i, z_j) = \frac{1}{2} \left[ \indyk (y_j \leq y_i)x_i + \indyk (y_i \leq y_j)x_j
 \right].
\end{equation}
To handle \eqref{derivQu} we will use tools for sums of independent random variables as well as the $U$-statistics theory. 
Namely, we  use exponential inequalities for sums of independent and unbounded random variables from \cite[Corollary 8.2]{geer_sparsity:16} and its version for $U$-statistics given in Lemma~D 
in the Appendix.   

\section{Simulation study}
\label{numerical}

In this section we present results of the comparative simulation study verifying the properties of RankLasso and its thresholded and adaptive versions in model selection.

 We consider the moderate dimension setup, where the number of explanatory variables $p$ increases with $n$ according to the formula $p = 0.01 n^2$. More specifically, we consider the following pairs $(n,p)$ : (100, 100), (200, 400), (300, 900), (400, 1600). For each of these combinations we consider three different values of the sparsity parameter $p_0=\#\{j: \beta_j \neq 0\}\in\{3,10,20\}$. 

In three of our simulation scenarios the rows of the design matrix are generated as independent random vectors from the multivariate normal distribution with the covariance matrix $\Sigma$ defined as follows\\
- for the {\it independent} case $\Sigma=I$,\\
- for the {\it correlated} case $\Sigma_{ii}=1$ and $\Sigma_{ij}=0.3$ for $i\neq j$.

In one of the scenarios  the design matrix is created by simulating the genotypes of $p$ independent Single Nucleotide Polymorphisms (SNPs).
In this case the explanatory variables can take only three values: 0 for the homozygote for the minor allele (genotype {\it \{a,a\}}), 1 for the heterozygote (genotype {\it \{a,A\}}) and 2 for the homozygote for the major allele  (genotype {\it \{A,A\}}). The frequencies of the minor allele for each SNP are independently drawn from the uniform distribution on the interval $(0.1, 0.5)$. Then, given the frequency $\pi_j$ for $j$-th SNP, the explanatory variable $X_{ij}$ has the distribution: $P(X_{ij}=0)=\pi_j^2$, $P(X_{ij}=1)=2\pi_j(1-\pi_j)$ and $P(X_{ij}=2)=(1-\pi_j)^2$.

The full description of the simulation scenarios is provided below:\\
- {\bf Scenario 1}
$$Y=X\beta+\varepsilon\;\;,$$
where $X$ matrix is generated according to the {\it independent} case, 
$\beta_1=\ldots=\beta_{p_0}=3$ and the elements of $\varepsilon= (\varepsilon_1,\ldots, \varepsilon_n)$ are independently drawn from the standard Cauchy distribution,\\
- {\bf Scenario 2} -
the regression model, values of regression coefficients and $\varepsilon$ are as in  Scenario 1, design matrix contains standardized versions of genotypes of $p$ independent SNPs, \\
- {\bf Scenario 3} -
the regression model, values of regression coefficients and $\varepsilon$ are as in  Scenario~1 and the design matrix $X$ is generated according to the {\it correlated} case,\\
- {\bf Scenario 4}  -
the design matrix $X$ is generated according to the {\it correlated} case and the relationship between $Y_i$ and $\beta 'X_i$ is non-linear:
$$Y_i=\exp(4+0.05 \beta ' X_i) + \varepsilon_i\;\;$$
and $\varepsilon_1,\ldots,\varepsilon_n$ are independent random variables from the standard Cauchy distribution.

In our simulation study we compare five different statistical methods:\\
- {\bf rL}: RankLasso defined in \eqref{Rlasso} with $\lambda:=\lambda_{rL}$ and 
\begin{equation}\label{lambda_rL}
\lambda_{rL}=0.3 \, \sqrt{\frac{\log p}{n}}\:.
\end{equation}
- {\bf arL}: adaptive RankLasso (\ref{gamma_a}), with $\lambda_a=2\lambda_{rL}$
and weights
$$w_j=\left\{\begin{array}{ccc}
\frac{0.1 \lambda_{rL}}{|\hth _j|}&\mbox{when}&|\hth _j|>0.1\lambda_{rL},\\
|\hth _j|^{-1}&\mbox{otherwise} \;\; ,&
\end{array}
\right.$$
where $\hth$ is the RankLasso estimator computed above.  If $\hth _j=0,$ then $|\hth _j|^{-1}=\infty$ and $j^{th}$ explanatory variable is removed from the list of predictors before running weighted RankLasso,\\
- {\bf thrL}: thresholded RankLasso, where the tuning parameter for RankLasso is selected by cross-validation and the threshold is selected in such a way that the number of selected predictors coincides with the number of predictors  selected by  adaptive RankLasso,\\
- {\bf LAD}: LADLasso, defined as 
$$  \arg \min_{\theta} \quad \frac{1}{n}\sum_{i=1}^n |Y_i-\theta ' X_{i}|+\lambda_{LAD} \sum_{j=1}^{p}|\theta_j|\;\;$$
with $\lambda_{LAD}=1.5\sqrt{\frac{\log p}{n}},$\\
- {\bf cv}:  regular Lasso with the tuning parameter selected by cross-validation.

The values of the tuning parameters for RankLasso and LADLasso were selected empirically so that both methods perform comparatively well for $p_0=3$ and $n=200, p=400$.

We compare the performance of the above methods by performing 200 replicates of the experiment, where in each replicate we generate the new realization of the design matrix $X$ and the vector of random noise $\varepsilon$. We calculate the following statistical characteristics:\\
- {\bf FDR}: the average value of $FDP=\frac{V}{\max(R,1)}$, where $R$ is the total number of selected predictors and $V$ is the number of selected irrelevant predictors, \\
- {\bf Power}:  the average value of $TPP=\frac{S}{p_0}$, where $S=R-V$ is the number of properly identified relevant predictors,\\
- {\bf NMP}: the average value of Numbers of Misclassified Predictors, i.e  false positives or false negatives, which equals  $V+p_0-S$.

In Table \ref{Table:time} we compare the average times needed to invoke RankLasso using the {\it glmnet} package \citep{glmnet:10} in the {\it R} software \citep{R:11} and the robust LADLasso using the {\it R} package {\it MTE} \citep{MTE2017}. It can be seen that LADLasso becomes prohibitively slow, when the number of columns of the design matrix exceeds 1000. For $p=1600$ 
computing the LADLasso estimator takes more than 30 seconds and is over 3000 times slower than calculating RankLasso.
\begin{table}[htbp]
{\footnotesize
	  \caption{Ratio of times needed to perform LADLasso (in {\it MTE}) and  RankLasso (in {\it glmnet})}
		\label{Table:time}
		\begin{center}
		\begin{tabular}{|c|c|c|c|}
		\hline
		dimension&t(LAD)/t(rL)\\
		\hline
		$n=100,p=100$&5.32\\
		$n=200,p=400$&95.2\\
		$n=300,p=900$&655\\
		$n=400,p=1600$&3087\\
		\hline
		\end{tabular}
		\end{center}
	}
\end{table}

Figure 1  illustrates the average number of falsely classified predictors for different methods and under different simulation scenarios. In the case when predictors are independent, RankLasso satisfies assumptions of \cite[Theorem 2.1]{WangZhu2015} and its NMP decreases with $p=0.01n^2$. The same is true for LADLasso. As shown in plots of FDR and Power (Figures 1 and 2 in the appendix), LADLasso becomes more conservative than RankLasso for larger values of $p_0$, which leads to slightly larger values of NMP. We can also observe that for independent predictors, the adaptive and thresholded versions perform similarly to the standard version of RankLasso.  As expected,  regular cross-validated Lasso performs very badly, when the error terms come from the Cauchy distribution. If suffers both from the loss of power and large FDR values. Also, it is interesting to observe that the first two rows in Figure 1 do not differ significantly, which shows that the performance of RankLasso for the realistic independent SNP data is very similar to its performance when the elements of the design matrix are drawn from the Gaussian distribution.

The behaviour of RankLasso changes significantly in the case when predictors are correlated.
Namely, NMP of RankLasso increases with $p$. On the other hand, NMP of both adaptive and thresholded versions of RankLasso decrease with $p$, so these two methods are able to find the true model consistently. 
In the case when the relationship between the median value of the response variable and the predictors is linear, LADLasso has a larger power and a similar FDR to RankLasso (see Figures \ref{Power} and \ref{FDR} in the appendix). In the result it returns substantially more of false positives and its NMP is slightly larger than the NMP of RankLasso. 
In the last row of Figure 1 we can observe that the lack of linearity has a negligible influence on the performance of RankLasso but substantially affects LADLasso, which now has a smaller power and  larger FDR than RankLasso.

\begin{figure}
\label{NMP}
\begin{center}
\includegraphics{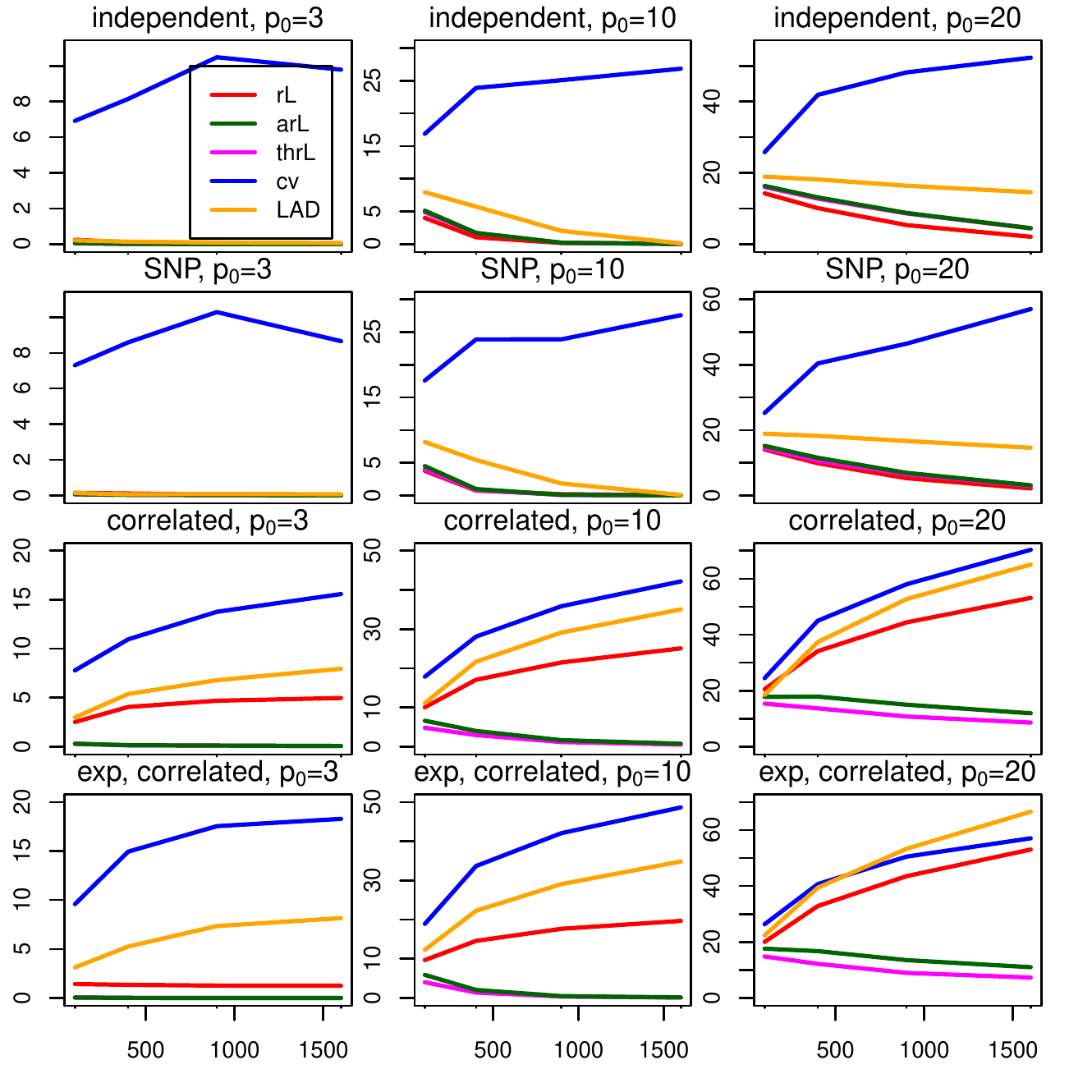}
\caption{\large{\bf \footnotesize Plots of NMP (average number of misclassified predictors) as the function of $p$.}}
\end{center}
\end{figure}

As shown in Figure 1, in the case of correlated predictors  thresholded RankLasso is systematically better than  adaptive RankLasso, even though both methods always select the same number of predictors. To explain this phenomenon, in Figure 2 we present the relationship between  the number of false discoveries (FD) and the number of true positives (TP) along the decreasingly ordered sequence of absolute values of RankLasso estimates for one realization of the experiment. We consider  RankLasso with the value of the tuning parameter selected as in (\ref{lambda_rL}) and by cross-validation (denoted by ,,cvrL'') and the adaptive RankLasso. 

Figure 2 illustrates substantial problems with the ordering of predictors by the estimates provided by RankLasso with the relatively large value of $\lambda$ from equation (\ref{lambda_rL}). Among three predictors with the largest absolute values of estimates of regression coefficiens, 2 are false positives.  The final estimate of RankLasso identifies 15 true predictors and returns 25 false discoveries.  Adaptive RankLasso does not provide a substantially better ordering of estimates of regression coefficients but stops at the point when the number of false discoveries drastically increases. In the result, adaptive RankLasso returns 7 true positives and 3 false discoveries. 
Comparing cross-validated RankLasso to its more conservative version, we can observe that it selects many more regressors -  it identifies all 20 of true predictors and returns 40 false discoveries. However, we can also observe that cross-validated RankLasso provides a much better ordering of the estimates of regression coefficients. The FD-TP curve of cross-validated RankLasso is substantially below the FD-TP curve corresponding to  RankLasso with  larger value of $\lambda$. Among 14 regressors with the largest absolute values of regression coefficient provided by  cross-validated RankLasso there is only one false positive. Thus, when selection is stopped by matching the adaptive RankLasso, we identify 9 true positives and only 1 false discovery.

\begin{figure}
\begin{center}
\includegraphics[scale=0.5]{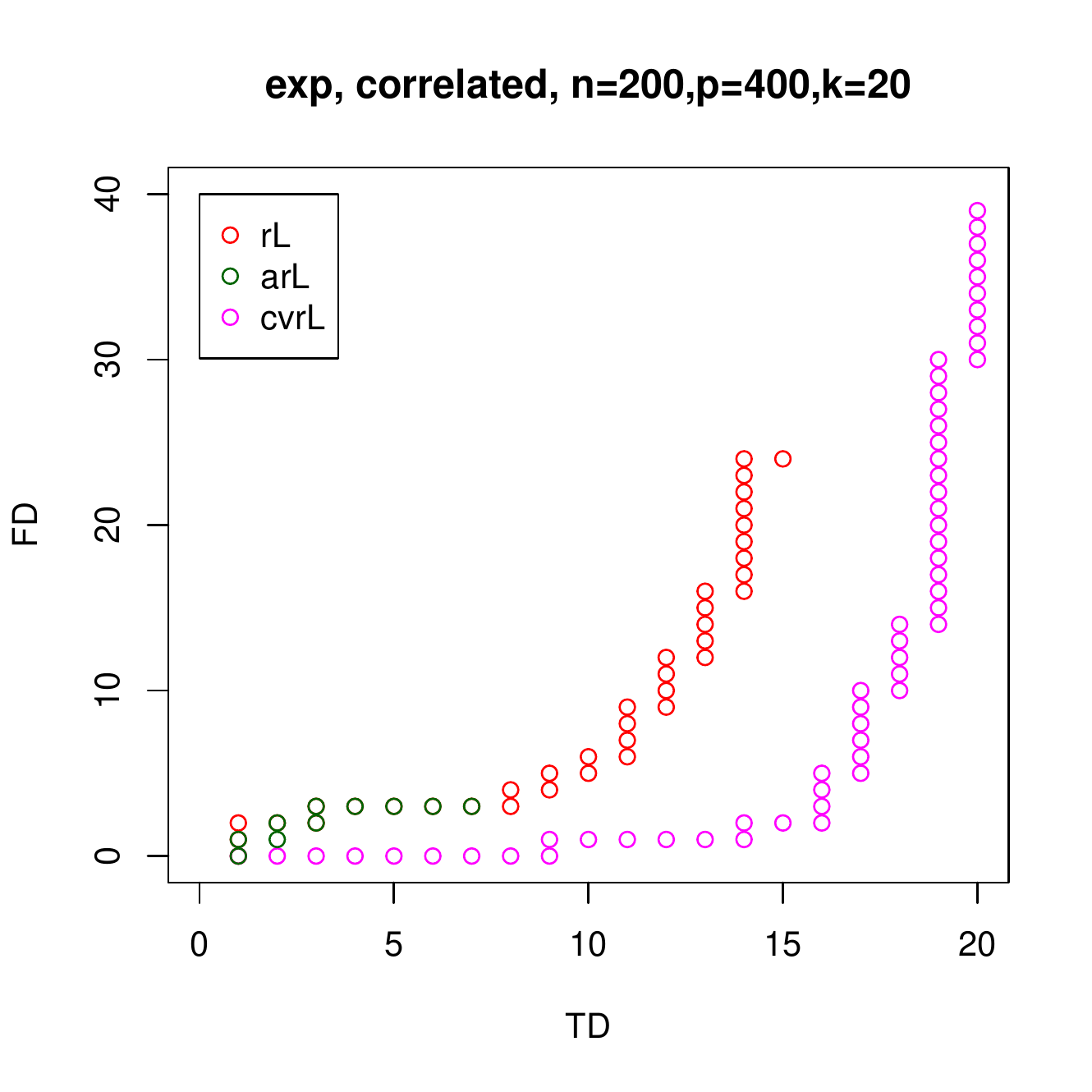}
\caption{\large{\bf \footnotesize Plots of the number of false discoveries (FD) vs the number of true positives (TP). 
}}
\end{center}
\label{Fig:FDP_TPP}
\end{figure}

\section{Analysis of real data}
\label{real}

In this section we apply rank methods and their competitors for identifying relationships between  expressions of different genes. Variations in gene expression levels may be related to phenotypic variations such as
susceptibility to diseases and response to drugs. Identifying the relationships between these gene expressions facilitates understanding the genetic-pathways and identifying regulatory genes influencing the disease processes. In the considered data set gene expression was interrogated in lymphoblastoid cell lines of 210 unrelated HapMap individuals \cite{hapmap2005} from four populations (60 Utah residents with ancestry from northern and western Europe, 45 Han Chinese in
Beijing, 45 Japanese in Tokyo, 60 Yoruba in Ibadan, Nigeria) \cite{stranger2007}. The data set can be found at {\it ftp://ftp.sanger.ac.uk/pub/genevar/} and was previously studied e.g. in \cite{bradic2011, fanfanbarut14}.
In our analysis we  will concentrate on four genes. First of them is the gene CCT8, which was analyzed previously in \cite{bradic2011}.  This gene is within the Down syndrome critical region on human chromosomen 21, on the minus strand. The over-expression of CCT8 may be associated with Down syndrome phenotypes. We also consider gene CHRNA6, which was previously investigated in \cite{fanfanbarut14} and is thought to be related to activation of dopamine releasing neurons with nicotine. 
Since the data on expression levels of these two genes contained only few relatively small outliers, 
we additionally considered  genes PRAME and Hs.444277-S, where the influence of outliers is more pronounced. The boxplots of these gene expressions can be found in Figure \ref{boxplots} in the appendix.

We start with preparing the data set using three pre-processing steps as in \cite{WangWuLi12}: 
we remove each probe for
which the maximum expression among 210 individuals is smaller than 
the 25-th percentile of the entire expression values,  we remove any probe for which the range of the expression among 210 individuals is smaller than 2 and finally we select 300 genes, whose expressions are  the most correlated to the expression level of the analyzed gene. 

Next, the data set is divided into two parts: the training set with randomly selected 180 individuals and the test set with remaining 30 individuals. Five procedures from Subsection \ref{numerical} are used to select important predictors using the training set and their accuracy is evaluated using the test set. As a measure of accuracy we cannot use the standard mean square prediction error, because in the single index model \eqref{model} the link funcion $g$ is unknown and the true parameter $\beta$ is not identifiable. However, 
we can expect that the ordering between values of the response variables $Y_i$ should be well predicted by the ordering between scalar products $\beta ' X_i.$  Moreover, 
from Theorem \ref{mult_cor} we know that $\theta ^0= \gamma _ \beta \beta  $ for the positive multiplicative number  $\gamma _ \beta , $ so the ordering between  $Y_i$ should be also well predicted by the ordering between   $(\theta {^0}) ' X_i.$
Therefore, as a accuracy measure of  estimators we use the {\it ordering prediction quality} (OPQ), which is defined as follows: let $T=n_t(n_t-1)/2$ be the number of different two-element subsets from the test set. 
The subset $\{i,j\}$ 
from the test set is properly ordered, if the sign of $ Y_i^{test} - Y_j^{test}$ coincides with the sign of  $\hat {\theta} 'X_i^{test} - \hat {\theta} 'X_j^{test},$ where $\hat \theta$ is some estimator of $\theta {^0}$ based on the training set. Let $P$ denote the number of two-element subsets from the test-set, that are properly ordered. The {\it ordering prediction quality} is defined as
\begin{equation}\label{OPQ}
OPQ=\frac{P}{T}\;\;.
\end{equation}

 Tables \ref{rp} and \ref{pred} report the average number of selected predictors and the average values of OPQ over 200 random splits into the training and the test sets. These values were calculated for all five  model selection methods  considered in the simulation study. 
\begin{table}[htb]
{\footnotesize
 \begin{center}
\caption{Average number of selected predictors (SP)}
\label{rp}
\begin{tabular}{||c|c|c|c|c|c||}
\hline\hline
SP&rL&arL&rLth&LADLasso&cv\\
\hline
CCT8&16&6&6&14&25\\
 CHRNA6&19&7&7&8&52\\
PRAME&16&6&6&6&0\\
Hs.444277-S&4&2&2&0&0\\
\hline
\hline
\end{tabular}
\end{center}
}
\end{table}

\begin{table}[htb]
{\footnotesize
\begin{center}
\caption{Average values of the Ordering Prediction Quality (OPQ, (\ref{OPQ}))}
\label{pred}
\begin{tabular}{||c|c|c|c|c|c||}
\hline\hline
OPQ&rL&arL&rLth&LADLasso&cv\\
\hline
CCT8&0.74&0.73&0.74&0.73&0.76\\
 CHRNA6&0.68&0.66&0.64&0.66&0.71\\
PRAME&0.65&0.62&0.61&0.60&0.08\\
Hs.444277-S&0.59&0.59&0.58&0.07&0.03\\
\hline
\hline
\end{tabular}
\end{center}
}
\end{table}
We can observe that for CCT8  the numbers of selected predictors and the prediction accuracy of  RankLasso and LADLasso are similar. The thresholded and adaptive RankLasso provide a similar prediction accuracy with  much smaller number of predictors. Interestingly,  regular cross-validated Lasso yields the best Ordering Prediction Quality, which however requires 4 times as many predictors as  thresholded or  adaptive RankLasso.  For CHRNA6 we see that the support of RankLasso is substantialy larger than for LADLasso and adaptive and thresholded RankLasso, which results in  slightly better prediction  accuracy. Again, the best prediction is obtained from  regular cross-validated Lasso, which however uses  much larger number of predictors. 

The performance of  regular Lasso drastically deteriorates for the remaining two genes, whose expressions contain substantially  larger outliers.
 Here  regular cross-validated Lasso in most of the cases is not capable of identifying any predictors. In the case of the gene Hs.444277-S the same is true about LADLasso. In the case of the PRAME gene,  the highest prediction accuracy is provided by  regular RankLasso, which however requires almost three times as many predictors as adaptive or thresholded Lasso. In the  case of Hs.444277-S RankLasso identifies 4 predictors, while its modified versions select only 2 genes. These simple models still allow to predict the ordering of gene expressions of Hs.444277-S with accuracy close to 60\%.

\vspace{0.3cm}

{\bf \large Acknowledgements}
We would like to thank Patrick Tardivel for helpful comments.  We gratefully acknowledge the grant of the Wroclaw Center of Networking and Supercomputing (WCSS), where most of the computations were performed.

\appendix

\vspace{3 cm}

{\Large \bf APPENDIX}

\vspace{1 cm}

In Section \ref{low_sec} of the appendix  we provide results for Rank-Lasso and its modifications in the low-dimensional scenario,
that we do not state in the main paper due to the space constraints. Besides, additional results of numerical experiments are in Section \ref{add_exp}. The proofs of results obtained in the main paper are given in Sections \ref{main_proofs1} and \ref{main_proofs2}. Finally, proofs of results from Section \ref{low_sec} in the appendix are given in Section     \ref{low_sec_proofs}.

\section{Low-dimensional scenario}
\label{low_sec}

In this section we consider properties of rank estimators in the case where the number of predictors is fixed. In the first part we focus on RankLasso and in the second part we study thresholded and weighted RankLasso.

We assume, without loss of generality,  that  $T=
\{1, \ldots, p_0\}$ for some  $0 < p_0 < p,$  so the response variable $Y$  depends on first $p_0$ predictors.
RankLasso estimates the set $T$  by
$$
\hat{T} = \{1\leq j \leq p: \hth _j \neq 0\}.
$$

 The results, that we obtain in this subsection, are asymptotic, so we can replace the true parameter $\t0$ in  \eqref{theta0_alt} by
\begin{equation}
\label{tstar}
\th*= \frac{n}{n-1} \t0 =H^{-1}\mu.
\end{equation}
 Obviously, it does not change the set of relevant predictors $T.$
We also decompose the matrix $H= \Ex X_1 X_1 '$ as
$$H= \left(
\begin{array}{c c}
\overbrace{H_1}^{p_0 \times p_0} & \overbrace{H_2}^{p_0 \times (p-p_0)} \\
H_2 ' & H_3
\end{array}
\right), $$
so the matrix $H_1$ describes correlations between relevant predictors and the matrix $H_2$ contains correlations between relevant and irrelevant predictors.

\subsubsection{Model selection consistency of RankLasso}
\label{unweight_low}

The next result provides sufficient and necessary conditions for RankLasso
to be model selection consistent. They are similar to the results proved in \cite[Theorem 1]{zou:06} and \cite[Theorem 1]{zhaoyu:06} that concern model selection in the linear model. Theorem \ref{nesscond} extends these results to the single index model \eqref{model}, which does not require any assumptions on the form of the link function nor the distribution of the noise variable.

\begin{theorem}
\label{nesscond}
Suppose that Assumption \ref{as_main} is satisfied, $\Ex |X_1|^4 < \infty$ and  $\lambda \rightarrow 0$, $\sqrt{n} \lambda \rightarrow \infty.$\\
(a)
If  $\lim _{n \rightarrow \infty}P(\hat{T} = T) \rightarrow 1,$ then
\begin{equation}
\label{ness}
 \left|
H_2 ' H_1^{-1} \sign(\th*_T) \right|_\infty \leq 1,
\end{equation}
where $\theta ^*$ is defined in \eqref{tstar}.\\
(b) If the inequality
\begin{equation}
\label{suff}
\left| H_2' H_1^{-1} \sign(\th*_T) \right|_\infty < 1
\end{equation}
holds, then $\lim _{n \rightarrow \infty}P(\hat{T} = T) \rightarrow 1$.
\end{theorem}

The sufficiency of \eqref{suff} for model selection consistency of RankLasso was established in \cite[Corollary 2.1]{WangZhu2015}. In Theorem \ref{nesscond} we strenghten this result by showing that it is almost the necessary condition.
The condition \eqref{suff}, called the irrepresentable condition \citep{zhaoyu:06}, is restrictive and satisfied only in some  very  special cases, like when  predictors are independent or when the correlations between ''neighboring'' variables decay exponentially with their distance.
Therefore,
RankLasso usually is not consistent in model selection.  However, as shown in the following Lemma \ref{pomoc}, it can consistently estimate $\th*$ under much weaker assumptions. This result
will be crucial for Subsection \ref{low_adap_lasso}, where we establish model selection consistency of the thresholded and weighted versions of RankLasso under such weaker assumptions. It is a generalization of \cite[Theorem 2]{kfu:00}.

\begin{lemma}
\label{pomoc}
Suppose that Assumption  \ref{as1} is satisfied and $\Ex |X_1|^4 < \infty.$ Let  $a_n$ be a sequence such that  $a_n \rightarrow 0, \; \frac{1}{a_n \sqrt{n}} \rightarrow b \in [0,\infty)$, $\frac{\lambda}{a_n }
\rightarrow c \in [0,\infty)$.  Then the RankLasso estimator  $\hat \theta$ in (\ref{Rlasso}) satisfies $$\frac{1}{a_n} \left(\hth-\th*\right) \rightarrow_d \arg \min\limits_\theta V(\theta),$$
where
$$V(\theta) = \frac{1}{2} \theta ' H \theta + b \,  \theta 'W + c\,\sum\limits_{j \in T} \theta_j \sign(\th*_j) + c\,\sum\limits_{j \notin T} |\theta_j|    $$
and $W$ has a normal $N(0,D)$ distribution with the matrix $D$ given in Lemma~\ref{asym}
in Section~\ref{low_sec_proofs}.
\end{lemma}

\subsubsection{Modifications of RankLasso}
\label{low_adap_lasso}

In this subsection we introduce two modifications of RankLasso and study their properties in the low-dimensional case.

First of these modifications, the weighted RankLasso estimator, is an analogue of the adaptive Lasso, proposed in  \cite{zou:06}.
The main idea of this approach relies on the application of different weights  for different predictors, depending on the value of  some initial estimator $\hbeta$ of  $\theta ^*$. This estimator needs to be $\sqrt{n}$-consistent, i.e. it needs to satisfy
\begin{equation}\label{rootn}
\sqrt{n} \left( \hbeta - \theta ^* \right)=O_P(1)\;\;.
\end{equation}
In particular, according to Lemma \ref{pomoc},  $\tilde{\theta}$ can be chosen as the RankLasso estimator with the regularization parameter
that behaves as $O (1/ \sqrt{n}).$
 Then, the weighted RankLasso estimator $\ahth$ is obtained  as
\begin{equation}\label{adlas}
\ahth=\arg
\min_{\theta \in \mathbb{R}^p}  \quad Q(\theta) + \lambda \sum_{j=1}^p w_j |\theta_j| \:,
\end{equation}
where $w_j = |\hbeta_j|^{-1}, j=1,\ldots,p$ and $ Q(\theta)$ is given in (\ref{Q_rand}).

Let $\hat{T}^a$ denote a set $\{j \in \{1, \ldots, p\}: \ahth _j \neq 0\}.$
The properties of $\ahth$ are described in the next theorem.

\begin{theorem}
\label{oracle_adap}
 Consider the weighted RankLasso estimator (\ref{adlas}) with $\hbeta$ satisfying (\ref{rootn}).
Suppose that Assumption \ref{as_main} is satisfied and $\Ex |X_1|^4 < \infty$. If $n \lambda \rightarrow \infty$ and $\sqrt{n} \lambda \rightarrow c \in [0,\infty),$ then \\
(a) $\lim\limits_{n \rightarrow \infty} \Pr \left(  \sign (\ahth ) = \sign(\beta ) \right)=1,$ where the equality of signs of two vectors is understood componentwise,\\
(b) $\sqrt{n}\left(\ahth_{T} - \th*_T \right) \rightarrow_d N \left(- H_1^{-1} \bar{c}, H_{1}^{-1} D_1 H_{1}^{-1} \right),$
where $\bar{c} = c \left(  \frac{1}{\th* _1},\ldots,  \frac{1}{\th* _{p_0}} \right),$ $\theta ^*$ is defined in \eqref{tstar}  and the matrix $D_1$ is the $(p_0 \times p_0)$ upper-left submatrix of
 the matrix $D$
defined in Lemma~\ref{asym} in Section \ref{low_sec_proofs}.
\end{theorem}

Now, we introduce the second modification, thresholded RankLasso. This estimator  is denoted by $\thth$ and defined in \eqref{thresh}.

\begin{theorem}
\label{asym_thresh}
Suppose that Assumption \ref{as_main} is satisfied and $\Ex |X_1|^4 < \infty$. If $\sqrt{n} \lambda \rightarrow  0,$ $\delta \rightarrow 0$ and $\sqrt{n} \delta \rightarrow \infty,$ then \\
(a) $\lim\limits_{n \rightarrow \infty} \Pr \left(  \sign (\thth ) = \sign(\beta ) \right)=1,$ where the equality of signs of two vectors is understood componentwise,\\
(b) $\sqrt{n}\left(\thth_{T} - \th*_T \right) \rightarrow_d N \left(0, \left(H^{-1} D H^{-1} \right)_1 \right),$
where  $\theta ^*$ is defined in \eqref{tstar} and $\left(H^{-1} D H^{-1} \right)_1$ is the $(p_0 \times p_0)$ upper-left submatrix of $H^{-1} D H^{-1}.$
\end{theorem}

Theorems \ref{oracle_adap} and  \ref{asym_thresh} state that   weighted and thresholded RankLasso behave almost like the oracle. They are asymptotically able to identify the support and recognize the signs of coordinates of the true parameter $\beta$.  Moreover, they estimate nonzero coordinates of $\theta ^*$ with the standard $\sqrt{n}$-rate. The crucial fact is that these theorems hold even when the irrepresentable condition is not satisfied. Thus, both modifications of RankLasso allow to identify the true model under much weaker assumptions than  vanilla RankLasso.


Theorems \ref{oracle_adap} and \ref{asym_thresh} work in the single index model \eqref{model} and they do not require  any assumptions on the distribution of the noise variables or the form of the increasing link function $g$.
 Comparing to other theoretical results concerning model selection with the robust loss functions, like \cite[Theorem]{WangLiJiang2007}, \cite[Theorem 2.1]{JohnsonPeng08},
\cite[Theorem 4.2]{songma:10}, \cite[Theorem 4.1]{Rejchel:17}, \cite[Theorem 2]{MedinaRon18},
 the assumptions of Theorems \ref{oracle_adap} and \ref{asym_thresh} are slightly stronger. Specifically, in  Theorems \ref{oracle_adap} and \ref{asym_thresh} the standard condition on the existence of the second moment of predictors is replaced by the assumption on the existence of the fourth moment.    This results from the fact that we work with the nonlinear model and the quadratic loss function. Apart from computational efficiency, application of the quadratic loss function allows us to solve the theoretical issues related to the dependency between ranks. The stronger assumption on the moments of predictors seems to be a relatively small prize for the gain in computational complexity, which allows to handle  large data sets. Moreover, according to the simulation study reported in Section \ref{numerical} of the main paper, for such large data sets  our method has substantially better statistical properties than LADLasso, which is a popular methodology for robust model selection.


\section{Results from Subsection \ref{unknown_link} of the main paper}
\label{main_proofs1}

Notice that for $Q(\theta)$ defined in \eqref{Q_rand} we have
$$
Q(\theta) = \frac{1}{2n} \sum_{i=1}^n \left(R_i/n   - \theta ' X_i \right)^2 + \theta ' \bar{X}/2-\frac{n+1}{4n} +1/8.
$$
Therefore, due to the fact that predictors $X_i$ are centred we will consider $Q(\theta)$ without subtracting $0.5, $ that is 
$$
Q(\theta) = \frac{1}{2n} \sum_{i=1}^n \left(R_i/n   - \theta ' X_i \right)^2 
$$
 in all proofs in this appendix. It will simplify notations.

\begin{proof}[Proof of Theorem \ref{mult_cor}]
We start with proving the first part of the theorem.
Argumentation is similar to the proof of \cite[Theorem 2.1]{LiDuan:89}, but it has to be adjusted to ranks which are not independent random variables (as distinct from $Y_1, \ldots, Y_n).$
Obviously, we have
$$
\Ex Q(\theta) = \frac{1}{2n^3} \sum_{i=1}^n \Ex R_i^2  -
 \frac{1}{n^2} \sum_{i=1}^n \Ex R_i \theta ' X_i +
 \frac{1}{2n} \sum_{i=1}^n \Ex \left(\theta ' X_i \right)^2.
$$
Vectors $(X_1, Y_1), \ldots, (X_n, Y_n) $ are i.i.d. and $X_i$ are centred, so  for all $i\neq 1$
\begin{eqnarray*}
&\,&\Ex R_{i}\theta ' X_i =
 \Ex \indyk (Y_1 \leq Y_i) \theta ' X_i + \sum_{j \neq \{1,i\}} \Ex \indyk (Y_j \leq Y_i) \theta ' X_i\\
&=& \Ex \indyk (Y_i \leq Y_1) \theta ' X_1 + \sum_{j \neq \{1,i\}} \Ex \indyk (Y_j \leq Y_1) \theta ' X_1 = \Ex R_1 \theta ' X_1.
\end{eqnarray*}
Moreover, ranks $R_1, \ldots, R_n$ have the same distribution, so $\sum_{i=1}^n \Ex R_i^2=n\Ex R_1^2.$ Therefore, we obtain that $\Ex Q (\theta) = \frac{1}{2}  \Ex \left(\frac{R_1}{n}  - \theta ' X_1 \right)^2.$ Using Jensen's inequality and Assumption \ref{as2}  we have
\begin{eqnarray*}
&\,&\Ex Q (\theta) = \frac{1}{2} \Ex \Ex \left[ \left(\frac{R_1}{n}  - \theta ' X_1 \right)^2
| \beta ' X_i, \varepsilon_i, i=1,\ldots, n\right]\\
&\geq& \frac{1}{2} \Ex \left[ \Ex \left( \frac{R_1}{n}  - \theta ' X_1
| \beta ' X_i, \varepsilon_i, i=1,\ldots, n\right) \right]^2 =
\frac{1}{2} \Ex \left[  \frac{R_1}{n}  - \Ex \left(\theta ' X_1
| \beta ' X_1\right) \right]^2\\
&=& \frac{1}{2} \Ex \left(  \frac{R_1}{n}  - d_\theta \beta ' X_1  \right)^2 \geq \min_{d \in \mathbb{R}} \Ex Q (d \beta).
\end{eqnarray*}
Obviously, we have $\min_d \Ex Q (d \beta) = \Ex Q(\gamma_\beta \beta),$ where $\gamma_\beta$ is defined in \eqref{gamma}.
Since $\theta^0$ is the unique minimizer of $\Ex Q(\theta)$, we obtain the first part of the theorem.

Next, we establish the second part of the theorem. Denote $Z=\beta ' X_1$ and $\varepsilon = \varepsilon _1.$ It is clear that $\gamma _\beta>0$ is equivalent to $Cov(Z, F(g(Z,\varepsilon)))>0.$ This covariance can be expressed as
\begin{equation}
\label{m1}
\Ex Z F(g(Z,\varepsilon)) = \Ex h(\varepsilon),
\end{equation}
where $h(a) = \Ex \left[Z F(g(Z,\varepsilon))|\varepsilon=a \right]= \Ex Z F(g(Z,a))$ for arbitrary $a.$ This fact simply follows from $\Ex Z =0$ and independence between $Z$ and $\varepsilon.$  If $F$ is increasing and $g$ is increasing with respect to the first variable, then $h(a)>0$ for arbitrary $a$ by Lemma \ref{covar_lemma} given below. Obviously, it implies that \eqref{m1} is positive.
\end{proof}

The following result was used in the proof of Theorem \ref{mult_cor}.  It is a simple and convenient adaptation of a well-known fact concerning covariance of nondecreasing functions \citep{Thorisson95}. Its proof follows \cite[Lemma A.44]{Kubkowski:19}.

\begin{lemma}
\label{covar_lemma}
Let U be a random variable that is not concentrated at one point, i.e. $P(U=u)<1$ for each $u \in \mathbb{R}.$ Moreover, let $f,h:\mathbb{R} \rightarrow \mathbb{R}$ be increasing functions. Then $Cov(f(U),h(U)) >0.$
\end{lemma}

\begin{proof}
For all real $a \neq b$ we have $[f(a)-f(b)][h(a)-h(b)] > 0,$ because $f,h$ are increasing. Let $V$ be an independent copy of $U.$ Then $P(U \neq V) >0$ and we obtain
\begin{eqnarray*}
0&<& \Ex [f(U)-f(V)][h(U)-h(V)] \,  \indyk (U \neq V) \\
&=& \Ex [f(U)-f(V)][h(U)-h(V)]\\
&=& 2 \Ex f(U) h(U) - 2 \Ex f(U) \, \Ex h(U)\\
&=& 2 Cov(f(U), h(U)).
\end{eqnarray*}
\end{proof}

\section{Results from Subsection \ref{high_sec} of the main paper}
\label{main_proofs2}

To prove Theorem \ref{high} we need three auxiliary results: Lemma \ref{geer_lemma}, Lemma \ref{termU} and Lemma \ref{termE}. The first one is borrowed from \cite[Corollary 8.2]{geer_sparsity:16}, while the second one is its adaptation to $U$-statistics.

\begin{lemma}
\label{geer_lemma}
Suppose that $Z_1, \ldots, Z_n$ are i.i.d. random variables and there exists $L>0$ such that $C^2= \Ex \exp\left( |Z_1|/L
\right)$ is finite. Then for arbitrary $u>0$
$$
P\left( \frac{1}{n} \sum _{i=1}^n (Z_i - \Ex Z_i) > 2L \left(
 C \sqrt{\frac{2u}{n}} + \frac{u}{n}
\right)
\right) \leq \exp(-u).
$$
\end{lemma}

\begin{lemma}
\label{termU}
Consider a $U$-statistic
$$
U = \frac{1}{n(n-1)} \sum_{i \neq j} h(Z_i, Z_j)
$$ with a kernel $h$ based on i.i.d. random variables $Z_1, \ldots, Z_n.$ Suppose that there exists $L>0$ such that $C^2= \Ex \exp\left( |h(Z_1,Z_2)|/L
\right)$ is finite. Then for arbitrary $u>0$
$$
P\left( U - \Ex U > 2L \left(
 C \sqrt{\frac{6u}{n}} + \frac{3u}{n}
\right)
\right) \leq \exp(-u).
$$
\end{lemma}

\begin{proof}
Let $g(z_1,z_2)= h(z_1,z_2) - \Ex h(Z_1,Z_2) $ and $\tilde{U}$ be a $U$-statistic with a kernel $g.$ Using Hoeffding's decomposition  we can represent every $U$-statistic as an
average of (dependent) averages of independent random variables \citep{serf:80}, i.e.
\begin{equation}
\label{Hdec}
\tilde{U} =
\frac{1}{n!} \sum_\pi \frac{1}{N} \sum_{i=1}^N
 g \left(Z_{\pi(i)}, Z_{\pi(N+i)} \right),
\end{equation}
where $N= \left\lfloor \frac{n}{2} \right\rfloor$ and  the first sum  on the right-hand side of \eqref{Hdec} is taken over all permutations $\pi$ of a set $\{1, \ldots, n\}.$
Take arbitrary $s>0.$ Then using Jensen's inequality and the fact that $Z_1, \ldots, Z_n$ are i.i.d.  we obtain
\begin{eqnarray}
\label{after_H}
\Ex \exp (s \tilde{U}) &\leq& \frac{1}{n!} \sum_\pi \Ex \exp \left[
\frac{s}{N} \sum_{i=1}^N
 g \left(Z_{\pi(i)}, Z_{\pi(N+i)} \right)
\right] \nonumber \\
&=& \Ex \exp \left[
\frac{s}{N} \sum_{i=1}^N
 g \left(Z_i, Z_{N+i} \right) \right].
\end{eqnarray}
We have average of $N$-i.i.d. random variables in \eqref{after_H}, so we can repeat argumentation from the proof of \cite[Corollary 8.2]{geer_sparsity:16}. Finally, we use a simple inequality $N\geq n/3$ for $n\geq 2.$
\end{proof}

\begin{lemma}
\label{termE}
Suppose Assumptions \ref{as_main} and \ref{as_sub} are satisfied. For arbitrary $j =1, \ldots,p$ and $u>0$ we have
\begin{equation}
\label{pro_termE}
P \left(\frac{1}{n}\sum_{i=1}^n X_{ij} X_i '\theta ^0 -\frac{n-1}{n}\mu_j > 5
\frac{\tau ^2}{ \sqrt{\kappa}}\left(
 2 \sqrt{\frac{2u}{n}} + \frac{u}{n}
\right)
\right)\leq \exp(-u).
\end{equation}
Besides, if $X_1$ has a normal distribution $N(0,H),$ then
we can  drop $\tau $ and $\kappa $ in \eqref{pro_termE}.
\end{lemma}

\begin{proof}
Fix $j =1, \ldots,p$ and $u>0.$ Recall that $H \theta^0 = \frac{n-1}{n} \mu $ by \eqref{theta0_alt}.
We work with an average of i.i.d. random variables, so  we can use Lemma \ref{geer_lemma}. We only have to find
$L, C>0$  such  that
$$
 \Ex \exp \left(
|X_{1j} X_1 '\theta ^0|/L
\right) \leq C^2.
$$
For each positive number $a,b,s$ we have the inequality $ab \leq \frac{a^2}{2s^2} + \frac{b^2s^2}{2}.$ Applying this fact and the Schwarz inequality we obtain
\begin{equation}
\label{eq2}
\Ex \exp \left(
|X_{1j} X_1 '\theta ^0|/L
\right) \leq \sqrt{\Ex \exp\left(\frac{X_{1j}^2}{s^2L}
\right) \Ex \exp\left(\frac{s^2 (X_1'\theta^0)^2}{L}
\right)}
\end{equation}
and the number $s$ will be chosen later. The variable $X_{1j}$ is subgaussian, so using \cite[Lemma 7.4]{baraniuk:11} we can bound the first expectation in \eqref{eq2} by $\left( 1- \frac{2 \tau ^2}{s^2L}\right)^{-1/2}$ provided that $s^2 L > 2 \tau ^2.$ The second expectation in \eqref{eq2} can be bounded using subgaussianity of the vector $X_1$ in the following way
$$
\Ex \exp\left(\frac{s^2 (X_1'\theta^0)^2  }{L}\right)  \leq
\left( 1- \frac{2 s^2 \tau ^2 |\theta^0|_2^2}{L}\right)^{-1/2},
$$
 provided that $2 s^2  \tau ^2 |\theta^0|_2^2< L.$
From Theorem \ref{mult_cor} we know that  $\theta^0 = \gamma_\beta \beta$ and
$\gamma _\beta =
\frac{\frac{n-1}{n} \, \Ex \,\indyk(Y_2 \leq Y_1) \beta ' X_1}{\beta ' H \beta}.
$
Recall that $\kappa $ is is the smallest eigenvalue of the matrix $H_T.$ Therefore,  we obtain a bound
$$
|\theta^0|_2^2 = \gamma_\beta ^2 |\beta_T|_2^2 \leq \kappa^{-1},
$$
because
$$
 \Ex \,\indyk(Y_2 \leq Y_1) \beta ' X_1 \leq  \sqrt{\beta_T ' H_T \beta_T}.
$$
Taking $L=2.2 \tau ^2 /\sqrt{\kappa}$ and $s^2= \sqrt{\kappa}$ we obtain $C \leq 2$ that finishes the proof of the first part of the lemma.

 Next, we assume that $X_1 \sim N(0,H).$ Therefore, $X_{1j} \sim N(0,1)$ and $(\theta^0)'X_1 \sim
N(0, (\theta^0)'H\theta^0).$ The argumentation is as above with $s^2=1.$  We only use the inequality $  (\theta^0)'H\theta^0 \leq 1$ and  the equality
$$
\Ex \exp \left((X_1'\theta^0)^2/L
\right)= \left(1-2 (\theta^0)'H\theta^0/L\right)^{-1/2},
$$
provided that $L> 2 (\theta^0)'H\theta^0.$ Therefore, we can take  $L=2.2.$

\end{proof}

\begin{lemma}
\label{infty}
Suppose that Assumption \ref{as_sub} and (\ref{n_cond}) are satisfied. Then for arbitrary $a \in (0,1), q \geq 1, \xi >1$ with probability at least $1-2a$ we have $\bfac \geq \fac /2.$
\end{lemma}

\begin{proof}

Fix $a \in (0,1),q \geq 1, \xi >1.$ We start with considering the $l_\infty$-norm of the matrix
$$
\left|\frac{1}{n} X' X - \Ex X_1 X_1 ' \right| _\infty = \max_{j,k=1,\ldots,p}
\left|\frac{1}{n} \sum_{i=1}^n X_{ij} X_{ik} - \Ex X_{1j} X_{1k}  \right|  .
$$
Fix $j,k \in \{1,\ldots, p\}. $  Using subgaussianity of predictors, Lemma \ref{geer_lemma} and argumentation similar to the  proof of Lemma \ref{termE} we have for $u = \log(p^2/a)$
$$
P\left(\left|\frac{1}{n} \sum_{i=1}^n X_{ij} X_{ik} - \Ex X_{1j} X_{1k}  \right|
> K_2 \tau^2 \sqrt{\frac{\log(p^2/a)}{n}}
\right) \leq \frac{2a}{p^2}\:,
$$
where $K_2$ is an universal constant. The values of constants $K_i$ that appear in this proof can change from line to line.

Therefore, using union bounds we obtain
$$
P\left(\left|\frac{1}{n}  X' X - \Ex X_{1} X_{1} ' \right|_\infty
> K_2 \tau^2 \sqrt{\frac{\log(p^2/a)}{n}}
\right) \leq 2a.
$$
Proceeding similarly to the proof of \cite[Lemma 4.1]{Cox13} we have the following probabilistic inequality
$$
\bfac \geq \fac - K_2 (1+\xi) p_0  \tau^2 \sqrt{\frac{\log(p^2/a)}{n}} \:.
$$
To finish the proof we use \eqref{n_cond} with $K_1$ being sufficiently large.
\end{proof}

\begin{proof}[Proof of Theorem \ref{high}]

Let $a \in (0,1)$ be arbitrary. The main part of the proof is to show that with high probability
\begin{equation}
\label{claim_rand}
|\hth - \theta^0|_q \leq\frac{2\xi p_0^{1/q}  \lambda}{(\xi +1)\bfac}\:.
\end{equation}
Then we  apply Lemma \ref{infty} to obtain  \eqref{high_estim}.

Thus, we focus on proving \eqref{claim_rand}. Denote $\Omega = \{|\nabla Q(\theta^0)|_\infty\leq \frac{\xi-1}{\xi+1} \lambda\}.$ We start with lower bounding  probability of $\Omega.$ For $A$ defined in \eqref{An} and every $j=1,\ldots,p $ we obtain
\begin{equation}
\label{deriv}
\nabla _j Q(\theta^0)= \left[\frac{1}{n}\sum_{i=1}^n X_{ij} X_i '\theta ^0 -\frac{n-1}{n}\mu_j \right] + \frac{n-1}{n} \left[\mu_j - A^j\right] - \frac{1}{n^2}\sum_{i=1}^n X_{ij} ,
\end{equation}
so if we find probabilistic bounds for each term on the right-hand side of \eqref{deriv}, then using union bounds we get the bound for $|\nabla Q(\theta^0)|_\infty. $
Consider the middle term in \eqref{deriv}. By \eqref{kernel} we apply Lemma \ref{termU} with  $h(z_1,z_2)=\frac{1}{2} \left[ \indyk (y_2 \leq y_1)x_{1j} + \indyk (y_1 \leq y_2)x_{2j} \right].$ Variables $X_{1j}$ and $X_{2j}$ are i.i.d., so  for arbitrary $L>0$  we have
\begin{equation}
\label{www}
\Ex \exp\left(|h(Z_1,Z_2| /L
\right) \leq \left[ \Ex \exp\left( |X_{1j}|/(2L)
\right)\right]^2.
\end{equation}
Using the fact that the variable $X_{1j}$ is subgaussian we bound
\eqref{www} by $4 \exp\left( \frac{\tau ^2}{4 L^2}
\right).$ Taking  $L =\tau $ and $u=\log (p/a)$ in Lemma \ref{termU} we obtain for some universal constant $K_1$
$$
P\left(A^j - \mu _j > K_1 \tau \sqrt{\frac{\log (p/a)}{n}} \right) \leq  \frac{a}{p}\:.
$$

The third term in \eqref{deriv} can be handled similarly using Lemma \ref{geer_lemma}. To obtain the bound for the first term in \eqref{deriv} we apply Lemma \ref{termE}.
Taking these results together and using union bounds we obtain that $P(\Omega) \geq 1- K_2 a$ provided that $\lambda$ satisfies \eqref{lambda}.

In further argumentation we consider only the event $\Omega.$ Besides,
we denote $\tth=\hth- \theta^0,$ where
$\hth$ is a minimizer of a convex function \eqref{Rlasso}, that is equivalent to
\begin{equation}
\label{min_equi}
\left\{
\begin{array}{ccc}
\nabla_j Q(\hth) = -\lambda \sign (\hth _j) &{\rm for }& \hth _j \neq 0,\\
|\nabla_j Q(\hth) | \leq \lambda  &{\rm for }& \hth _j = 0,
\end{array}
\right.
\end{equation}
where $j=1,\ldots,p.$

First, we prove that $\tth \in \cone.$ Here our argumentation is standard \citep{yezhang:10}.
 From \eqref{min_equi} and the fact that  $|\tth|_1 = |\tth _T|_1 +|  \tth _{T'}|_1$ we can calculate
\begin{align*}
0 &\leq \tth ' X'X \tth /n  = \tth '\left[  \nabla Q (\hth) - \nabla Q (\theta^0)\right]\\
&= \sum_{j \in T} \tth _j \nabla _j Q(\hth) +
\sum_{j \in T'} \hth _j \nabla _j Q(\hth) - \tth ' \nabla Q (\theta^0) \\
&\leq  \lambda \sum_{ j \in T} |\tth _j | - \lambda \sum_{j \in T'} |\hth _j |
+|\tth|_1 |\nabla Q(\theta^0)|_\infty \\
&= \left[\lambda +|\nabla Q(\theta^0)|_\infty \right] |\tth _T|_1 + \left[ |\nabla Q(\theta^0)|_\infty- \lambda \right] |\tth _{T'}|_1\, .
\end{align*}
Thus, using the fact that we consider the event $\Omega$ we get
\[
|\tth _{T'}|_1 \leq \frac{\lambda+|\nabla Q(\theta^0)|_\infty}{\lambda-|\nabla Q(\theta^0)|_\infty} |\tth _T|_1 \leq \xi |\tth _T|_1\, .
\]
Therefore, from the definition of $\bfac$ we have
$$
 |\hth - \theta^0|_q \leq\frac{p_0^{1/q} |X'X(\hth - \theta^0)/n|_\infty}{\bfac}
\leq p_0^{1/q} \frac{|\nabla Q(\hth)|_\infty +| \nabla Q(\theta^0)|_\infty}{\bfac} \:.
$$
Using \eqref{min_equi} and the fact, that we are on $\Omega ,$ we obtain
\eqref{claim_rand}.

The case $X_1 \sim N(0,H)$ is a consequence of the analogous part of Lemma~\ref{termE}.
\end{proof}

\begin{proof}[Proof of Corollary \ref{seperation}]
The proof is a simple consequence of the bound \eqref{high_estim} with $q=\infty$ obtained in Theorem \ref{high}. Indeed, for  arbitrary predictors $j \in T$ and $k \notin T$ we obtain
$$
|\hth _j | \geq  |\t0 _j| -|\hth _j  - \t0 _j| \geq \t0 _{min} - |\hth - \t0|_\infty >
\frac{4 \xi  \lambda}{ (\xi +1) F _ \infty (\xi)}
 \geq
|\hth _k  - \t0 _k| = |\hth _k  |.
$$

\end{proof}

\begin{proof}[Proof of Theorem \ref{th_high}]
The proof is a simple consequence of the uniform bound \eqref{bound} from Corollary \ref{cor_simply}. Indeed, for an  arbitrary $j \notin T$ we obtain
\[
|\hth _j | =|\hth _j  - \t0 _j| \leq K_4 \lambda < \delta\,,
\]
so $j \notin \hat{T}^{th}.$
Analogously, if  $j \in T,$ then
\[
|\hth _j | \geq  |\t0 _j| -|\hth _j  - \t0 _j| \geq 2 \delta  - K_4 \lambda > \delta\,.
\]
\end{proof}

\begin{proof}[Proof of Theorem \ref{adap_high}]
First, we define a function
\begin{equation}
\label{gamma_a_proof}
\Gamma^a(\theta) = Q(\theta) + \lambda_a \sum_{j=1}^p w_j |\theta_j| .
\end{equation}
Next, we fix $a \in (0,1)$ and set, for simplicity, $\xi_0 =3.$
 Consider the event
$\Omega = \{|\nabla Q(\theta^0)|_\infty\leq  \lambda/2 \}.$ We know from the proof of Theorem \ref{high} that $P(\Omega) \geq 1-K_3a$ and the inequality \eqref{bound} is satisfied. The proof of Theorem \ref{adap_high} consists of two steps. In the first one we show that with high probability there exists a minimizer of the function
$$
g(\theta_T) = \Gamma^a (\theta_T,0)
$$
that is close to $\theta^0_T$ in the $l_1$-norm. We denote this minimizer by $\hth ^a _T.$ In the second part of the proof we obtain that the vector $(\hth ^a _T, 0),$ that is $\hth ^a _T$ augmented by $(p-p_0)$ zeros, minimizes the function
\eqref{gamma_a_proof}.

First, consider vectors $v \in \mathbb{R} ^{p_0}$ having a fixed common $l_1$-norm and a sphere
\begin{equation}
\label{sphere}
\{ \theta_T= \theta^0_T +p_0 \lambda v\}.
\end{equation}
Suppose that $|v|_1$ is sufficiently large. We take arbitrary $\theta_T$ from the sphere \eqref{sphere} and calculate that
$$
Q(\theta_T,0)- Q(\theta^0) =\frac{1}{2} p_0^2\lambda^2 v' \frac{1}{n} X_T'X_T v+
p_0 \lambda v'[\nabla Q(\theta^0)]_T.
$$
Let $\hat{\kappa}$ be the minimal eigenvalue of the matrix $\frac{1}{n} X_T'X_T.$
Then we have $v' \frac{1}{n} X_T'X_T v \geq \hat{\kappa} |v|_1^2/p_0.$ Besides, on the event $\Omega$ we obtain
$$
 |v'[\nabla Q(\theta^0)]_T| \leq |v|_1  |[\nabla Q(\theta^0)]_T|_\infty \leq  \lambda |v|_1/2.
$$
Proceeding analogously to the proof of Lemma \ref{infty} we can show that  $\hat{\kappa} \geq \kappa /2$ with probability close to one. Therefore, we obtain
\begin{equation}
\label{lowQ}
Q(\theta_T,0)- Q(\theta^0) \geq \kappa p_0\lambda^2 |v|_1^2/4-
 p_0 \lambda^2  |v|_1/2.
\end{equation}
Next, we work with the penalty term and obtain
\begin{equation}
\label{penal}
\left|\lambda_a \sum_{j=1}^{p_0} w_j \left[|\theta^0_j +p_0 \lambda v_j|- |\theta^0_j|  \right]
\right| \leq
\lambda_a p_0 \lambda \sum_{j=1}^{p_0} w_j  |v_j|.
\end{equation}
Moreover, for $j \in T$ we have from Corollary \ref{cor_simply} that
$$
|\hth _j| \geq |\theta ^0_j| - |\hth _j - \theta^0_j|\geq \thmin - K_4 \lambda >\lambda_a,
$$
so $w_j  \leq K.$ Therefore, the right-hand side of\eqref{penal} is bounded by $K \lambda \lambda_a p_0 |v|_1.$ Combining it with \eqref{lowQ} we get
\begin{equation}
\label{bound_g}
g(\theta_T)- g(\theta^0_T) \geq p_0 \lambda^2 |v|_1 \left(\kappa  |v|_1/4- 1/2  - K_4  K  \right).
\end{equation}
The right-hand side of \eqref{bound_g} is positive, because the norm $|v|_1$ can be taken sufficiently large, $ K, K_4$ are constants and $\kappa$ is lower bounded by a constant. Therefore, the convex function $g(\theta_T)$ takes on a sphere \eqref{sphere} values larger than in the center $\theta^0_T.$ So, there exists a minimizer inside this sphere.

Next, we show that the random vector $(\hth^a_T,0)$ minimizes \eqref{gamma_a_proof} with high probability, so we have to prove that the event
\begin{equation}
\label{event1}
\{|\nabla_j Q(\hth ^a_T,0)| \leq w_j \lambda_a \quad {\rm for \; every } \; j \notin T
\}
\end{equation}
has probability close to one.  By Corollary \ref{cor_simply}
we have for $j \notin T$
$$
|\hth _j| = |\hth _j - \theta_j ^0| \leq K_4 \lambda.
$$
Therefore, the corresponding weight $w_j \geq \lambda_a^{-1}.$ We can also calculate that
$$
\nabla Q(\theta_T,0) = \frac{1}{n} X'X_T \theta_T - \left[\frac{n-1}{n}A +\frac{1}{n^2}
\sum_{i=1}^nX_i
\right],
$$
so we obtain the inequality
\begin{equation}
\label{ineq1}
\left| \left[ \nabla Q(\hth ^a_T,0) \right]_{T'}\right|_\infty \leq \left|
\frac{1}{n} X_{T'}'X_T (\hth ^a_T-\theta^0_T)
\right|_\infty +\left| \left[ \nabla Q(\theta^0) \right]_{T'} \right|_\infty.
\end{equation}
Consider the event $\Omega = \{|\nabla Q(\theta^0)|_\infty\leq  \lambda/2 \}$ that has probability close to one by the proof of Theorem \ref{high}.  Then the second term on the right-hand side of \eqref{ineq1} can be bounded by $\lambda/2.$ The former one can be decomposed as
\begin{eqnarray}
\label{ineq2}
\left|\frac{1}{n} X_{T'}'X_T (\hth ^a_T-\theta^0_T) \right|_\infty &\leq& \left|
\left(\frac{1}{n} X_{T'}'X_T - H_2' \right) (\hth ^a_T-\theta^0_T)
\right|_\infty + \left|
H_2' (\hth ^a_T-\theta^0_T)
\right|_\infty \nonumber \\
&\leq &
\left|
\frac{1}{n} X_{T'}'X_T - H_2' \right|_\infty \left|\hth ^a_T-\theta^0_T \right|_1 + \left|
H_2' \right|_\infty \left|\hth ^a_T-\theta^0_T \right|_1.
\end{eqnarray}
The expression $|H_2|_\infty$ is bounded by one, so from the first part of the proof we can bound, with high probability, the second term in \eqref{ineq2} by $K_6 p_0 \lambda.$
The $l_\infty$-norm in the former expression can be bounded, with probability close to one, by $K_7 \sqrt{\frac{\log (p/a)}{n}}$ as in the proof of Lemma \ref{infty}. Therefore, we have just proven that with probability close to one
$$
\left| \left[ \nabla Q(\hth ^a_T,0) \right]_{T'}\right|_\infty \leq K_8 p_0 \lambda.
$$
Combining it with the fact that $w_j \geq \lambda_a^{-1}$ we obtain that the event \eqref{event1}  has probability close to one, because  from assumptions of the theorem $ p_0 \lambda \leq K_5$
for $K_5$ small enough.
\end{proof}

\section{Results from Section \ref{low_sec}}
\label{low_sec_proofs}

We start with the proof of Lemma~\ref{pomoc}. Then we state Lemma \ref{asym} that is also
needed in proofs of Theorems \ref{nesscond},  \ref{oracle_adap} and \ref{asym_thresh}. We will use the following notation
\begin{equation}
\label{ggamma}
\Gamma (\theta)=Q(\theta) + \lambda |\theta|_1.
\end{equation}

\begin{proof}[Proof of Lemma \ref{pomoc}]
Let $a:=a_n  $ be a fixed sequence such that $a \rightarrow 0.$ We can calculate that for every $\theta$
\begin{equation}
\label{num1}
Q(\th*+a \theta) - Q(\th*) = -\frac{a}{n^2} \theta '\left( \sum_{i=1}^n R_i X_i \right) + a \theta ' \left(\frac{X ' X}{n} \right) \th* + \frac{a^2}{2} \theta '
\left(\frac{X ' X}{n} \right)\theta .
\end{equation}
Using \eqref{sum_rank} we obtain that the right-hand side of \eqref{num1} is
$$
\frac{a^2}{2} \theta '
\left(\frac{X ' X}{n} \right)\theta -a \theta ' \left[
\frac{n-1}{n} A + \bar{X}/n - \left(\frac{X ' X}{n} \right) \th*
\right],
$$
where $A$ is defined in \eqref{An}.
Therefore, we have
\begin{eqnarray*}
&\,&\frac{1}{a^2} \left[ Q (\th* + a \theta)-  Q(\th*) \right]=
\frac{1}{2} \theta '
\left(\frac{X ' X}{n} \right)\theta \\&-& \frac{\theta '}{\sqrt{n} a}
 \left[
\frac{n-1}{n} \sqrt{n}A + \bar{X}/\sqrt{n} -  \left(\frac{X ' X}{\sqrt{n}} \right) \th*
\right].
\end{eqnarray*}
Using LLN, Lemma \ref{asym} (given below) and Slutsky's theorem we get that
\begin{equation}
\label{conv11}
\frac{1}{a^2} \left[ Q (\th* + a \theta)-  Q(\th*) \right]
\rightarrow_{f-d}\frac{1}{2} \theta ^T H \theta + b \theta ' W ,
\end{equation}
where $\rightarrow_{f-d}$ is the finite-dimensional convergence in distribution and $W \sim
N(0, D).$
Next, we consider the penalty term and notice that
\begin{equation}
\label{regul}
\frac{\lambda}{a^2} \sum_{j=1}^p \left( | \th*_j + a \theta_j |
- \left| \th*_j \right|  \right) \rightarrow c \sum\limits_{j \in T} \theta_j \sign(\th*_j) + c \sum\limits_{j \notin T} |\theta_j| .
\end{equation}
Thus, from (\ref{conv11}) and (\ref{regul}) we have the convergence of convex functions
\begin{equation}
\label{conv2}
\frac{1}{a^2} \left[ \Gamma(\th* + a \theta)-  \Gamma(\th*) \right]
 \rightarrow_{f-d}  V(\theta) ,
\end{equation}
where the function $\Gamma (\theta)$ is defined in \eqref{ggamma}.
The function on the left-hand side of (\ref{conv2}) is minimized by $\frac{1}{a} \left(\hth - \th*\right)$ and the convex function on the right-hand side of (\ref{conv2}) has a unique minimizer. Thus $\frac{1}{a} \left(\hth - \th*\right) \rightarrow_d \arg \min\limits_\theta V(\theta)$
\citep[see][]{geyer:96}.
\end{proof}

\begin{lemma}
\label{asym}
Suppose that Assumption \ref{as1} is satisfied and $\Ex |X_1|^4 < \infty$. Then
$$
\sqrt{n} \left[A - \mu \right] - \sqrt{n} \left[\frac{X'X}{n} \th* - \mu \right]
\rightarrow_d N(0,D),
$$
where $D$ is stated precisely in the proof below.
\end{lemma}

\begin{proof}
Consider two $U$-statistics. The first one is $A$ that is defined in \eqref{An}.
 The second $U$-statistic is
$$
B=\frac{1}{n}\sum_{i=1}^n X_i X_i'\th*
$$
and is of the order one. Besides, we have $\Ex B = H \th* =\mu$ by \eqref{tstar}.
Using \cite[Theorem 7.1]{hoeff:48} we obtain convergence in distribution  in $\mathbb{R}^{2p}$
$$
\sqrt{n} \left[
\begin{array}{c}
A-\mu\\
B-\mu
\end{array}
\right]
\rightarrow_d N(0,\Sigma)
$$
for the matrix
$$\Sigma= \left(
\begin{array}{c c}
\overbrace{\Sigma_1}^{p \times p} & \overbrace{\Sigma_2}^{p \times p} \\
\Sigma_2  & \Sigma_3
\end{array}
\right), $$
where for $j,k=1,\ldots ,p$ and the function $f$ in \eqref{kernel} we have
$$\left(\Sigma_1\right)_{jk} = 4 Cov(\tilde{f}_j(Z_1),\tilde{f}_k(Z_1)),$$
where $\tilde{f}(z_1) = \Ex \left[  f(Z_1,Z_2)|Z_1=z_1\right]$ and $\tilde{f}_j(z_1)$ is its
$j$-th coordinate. The entries of the matrix $\Sigma _3$ are
$$\left(\Sigma_3\right)_{jk} = Cov(X_{1j}X_1'\th*, X_{1k}X_1'\th*)$$
and
$$\left(\Sigma_2\right)_{jk} = 2Cov(\tilde{f}_j(Z_1), X_{1k}X_1'\th*).$$
Next, define $(p\times 2p)$-dimensional matrix $M$ in the following way: for $j=1,\ldots,
p$ put  $M_{j,j}=1$ and $M_{j,p+j}=-1,$ and zeros elsewhere. Then
$$
\sqrt{n} \left[A - \mu \right] - \sqrt{n} \left[\frac{X'X}{n} \th* - \mu \right] =
M \sqrt{n} \left[
\begin{array}{c}
A-\mu\\
B-\mu
\end{array}
\right]\rightarrow_d N(0,M \Sigma M') .
$$
\end{proof}

Now we prove  main results of Section \ref{low_sec}.

\begin{proof}[Proof of Theorem \ref{nesscond}]
From Lemma \ref{pomoc} for $a=\lambda$ we obtain
\begin{equation}
\label{conv}
\lambda^{-1} \left(\hth-\th*\right) \rightarrow_d \arg \min\limits_\theta V_2(\theta),
\end{equation}
where
$$V_2(\theta) = \frac{1}{2} \theta ' H \theta  + \sum\limits_{j \in T} \theta_j \sign(\th*_j) + \sum\limits_{j \notin T} |\theta_j| .   $$
The proof of the claim (a) is similar to the proof of \cite[Theorem 1, scenario (3)]{zou:06} and uses properties of the function $V_2(\theta).$ Therefore, we consider only the case (b).
Let $\eta = \arg \min_\theta V_2(\theta).$ We know that $\eta$ is nonrandom and the function $V_2(\theta) $ is strictly convex. Therefore, using \eqref{suff} we have
\begin{equation}
\label{meta}
\eta = \left( - H_1^{-1} \sign (\th* _T), 0\right).
\end{equation}
For fixed $ j \in T$ we have
$$\lambda^{-1} \left(\hth _j-\th* _j\right) \rightarrow_P \eta_j,$$
so
$
\Pr( j \notin \hat{T}) = \Pr (\hth _j=0) \rightarrow 0.
$
Thus,
$\Pr (T \subset \hat{T}) \rightarrow 1.$
Next, we show that
$\Pr (\hat{T} \subset T)
 \geq 1- \sum\limits_{j \notin T}
\Pr(j \in \hat{T}) \rightarrow 1.$
Consider  fixed $ j \notin T $ and an event $\{j \in \hat{T}  \}$.
Recall that $\hth $ minimizes the convex function $\Gamma ,$ so $0 \in \partial \Gamma
(\hth )$, where $\partial \Gamma$ denotes a subgradient of the convex function $\Gamma .$ The function $Q(\theta)$ is differentiable, so
$
\partial \Gamma  (\hth )= \nabla Q(\hth ) + \lambda
\partial |\hth  | .
$
Therefore, we have
\begin{equation}
\label{suff_form2}
0=  \nabla^j Q(\hth ) +
\lambda
\: \sign (\hth_j ) ,
\end{equation}
where $\nabla ^j Q(\hth)$ is the $j$-th partial derivative $Q(\theta)$ at $\hth$. From (\ref{suff_form2}) we have
\begin{equation}
\label{suff_form3}
\lambda^{-1 }   \left|\nabla^j Q(\hth ) \right| =1.
\end{equation}
We can calculate that
$$
\nabla Q(\hth) = -\frac{n-1}{n}A- \bar{X}/n +  \frac{X ' X}{n} \left(\hth -\th*\right)
+\frac{X ' X}{n} \th*,
$$
that gives us
\begin{eqnarray}
\label{suff_form4}
\lambda^{-1}  \,  \nabla Q(\hth )  &=&
-\frac{1}{\sqrt{n}\lambda }  \, \left[ \frac{n-1}{n} \sqrt{n} A+ \bar{X}/\sqrt{n}-
\sqrt{n} \frac{X ' X}{n} \th*\right] \nonumber \\
&+&
\lambda^{-1 }  \frac{X ' X}{n} (\hth- \th*).
\end{eqnarray}
Therefore, using LLN, Lemma \ref{asym}, \eqref{conv} and Slutsky's theorem  the left-hand side of (\ref{suff_form3}) tends in probability to $\left| \left(H \eta\right)_j
\right|.$
Recall that we consider the event $\{j \in \hat{T}  \}$ for $j \notin T,$ so we have the inequality
$$
\limsup _{n \rightarrow \infty} \Pr \left( \hth_j \neq 0 \right)\leq  \indyk \left(\left| \left(H \eta\right)_j
\right| =1 \right),
$$
since $\eta $ is not random. However, from \eqref{suff}  and \eqref{meta} we obtain
 $$ \left| \left(H \eta\right)_j \right| = \left| \left( H_2 ^T H_1^{-1} \sign
\left( \th* _\mA \right) \right)_j \right| <1.$$ Therefore, probability
$\Pr \left( \hth_j \neq 0 \right)$ tends to zero that finishes the proof of consistency in model
selection.
\end{proof}

\begin{proof}[Proof of Theorem \ref{oracle_adap}]
We define a function $$
\Gamma^a (\theta)=Q(\theta) + \lambda \sum_{j=1}^p \frac{|\theta_j|}{|\hbeta_j|} \:.
$$
Let us start with the claim (b). Repeating the same arguments as in the proof of Lemma \ref{pomoc} (for  $a = \frac{1}{\sqrt{n}} $) we obtain for every $\theta$
$$
n Q\left(\th* + \frac{\theta}{\sqrt{n}}\right)- n Q(\th*)
\rightarrow_{f-d} \frac{1}{2} \theta ' H \theta+\theta ' W ,
$$
which using convexity implies weak convergence of the stochastic process
\begin{equation}
\label{oracff}
\left\{n Q\left(\th* + \frac{\theta}{\sqrt{n}}\right)- n Q(\th*) : \theta \in K\right\}
\rightarrow_d \left\{\frac{1}{2} \theta ' H \theta+\theta ' W : \theta \in K \right\}
\end{equation}
for every compact set $K$ in $\mathbb{R} ^p $  \citep[see][]{ arcones:98}.
Now consider the penalty term and use similar arguments to that in the proof of
\cite[Theorem 2]{zou:06} to obtain that if $\th*_j \neq 0,$ then
$$
n\lambda \left( \frac{\left|\th*_j + \frac{\theta_j}{\sqrt{n}}\right| -
|\th*_j| }{|\hbeta_j|} \right) = \sqrt{n} \lambda  \sqrt{n} \left( \frac{\left|\th*_j + \frac{\theta_j}{\sqrt{n}}\right| -
|\th*_j| }{|\hbeta_j|} \right) \rightarrow_P c \frac{\theta_j}{\th* _j},
$$
because $\sqrt{n}\lambda \rightarrow c, \hbeta _j \rightarrow_P \th* _j $ and $\sqrt{n} \left[\left|\th*_j + \frac{\theta_j}{\sqrt{n}}\right| -
|\th*_j| \right] \rightarrow \sign(\th*_j) \theta_j.$
However, if $\th*_j = 0,$ then
$$
n\lambda \left( \frac{\left|\th*_j + \frac{\theta_j}{\sqrt{n}}\right| -
|\th*_j| }{|\hbeta_j|} \right) = \sqrt{n}\lambda \frac{|\theta_j|}{ |\hbeta_j|}  \rightarrow_P \left\{
\begin{array}{cc}
0, & \theta_j= 0\\
\infty, & \theta_j \neq  0,
\end{array}
\right.
$$
since $\sqrt{n} \hbeta _j=O_P(1)$ and $n \lambda \rightarrow \infty$. Therefore, we
obtain that for every $\theta$
$$
n \lambda \sum_{j=1}^p \frac{\left|\th*_j + \frac{\theta_j}{\sqrt{n}}\right| -
|\th*_j| }{|\hbeta_j|}    \rightarrow_P \left\{
\begin{array}{cc}
c \sum_{j \in T} \frac{\theta_j}{\th* _j}, & \theta = \left(\theta_1, \ldots, \theta_{p_0} , 0,\ldots, 0\right) \\
\infty, &  \rm{otherwise.}
\end{array}
\right.
$$
Since we have infinity in the last limit we cannot use arguments based on uniform convergence on compacts as we have done in the proof of Lemma \ref{pomoc}. Here we should follow epi-convergence results \citep{geyer:94, pflug:95, zou:06} that combined with
convergence (\ref{oracff}) give us
\begin{equation}
\label{epi-c}
nQ \left(\th* +\frac{\theta}{\sqrt{n}}\right) - n Q (\th*)+ n \lambda \sum_{j=1}^p \frac{\left|\th*_j + \frac{\theta_j}{\sqrt{n}}\right| -
|\th*_j| }{|\hbeta_j|} \rightarrow _{e-d} V_3(\theta),
\end{equation}
where
$$
V_3(\theta)=\left\{
\begin{array}{cc}
\frac{1}{2} \theta_T ' H_{1} \theta_T + \theta_T ' (W_T + \bar{c}),&
\left(\theta_1, \ldots, \theta_{p_0} , 0,\ldots, 0\right)\\
\infty, & \rm{otherwise}
\end{array}
\right.
$$
and $W_T \sim N(0,D_1).$ Convergence $\rightarrow_{e-d}$ is epi-convergence in
distribution \citep{geyer:94, pflug:95}.
Furthermore, the function $V_3$ has the unique minimizer $\left[ -H_{1}^{-1} (W_T + \bar{c}), 0 \right] ' ,$  so
epi-convergence in (\ref{epi-c}) implies convergence of minimizers \citep[see][]{geyer:94}
\begin{equation}
\label{end_form}
\sqrt{n}\left(\ahth_{T} - \th*_T \right) \rightarrow_d - H^{-1}_{1} (W_T + \bar{c}) \quad \quad
{\rm and} \quad \quad \sqrt{n}\left(\ahth_{T '} - \th*_{T '} \right) \rightarrow_d 0,
\end{equation}
where $T'=\{p_0+1, \ldots, p\}$ is the complement of $T.$
It finishes the proof of the second claim.

Next, we go to the claim (a). We prove only that
$$\lim\limits_{n \rightarrow \infty} \Pr \left( \hat{T}^a = T\right)=1,$$
because the equality of the signs of relevant predictors follows simply from Theorem \ref{mult_cor} and  estimation consistency stated in the claim (b) of this theorem.
The reasoning is similar to the proof of Theorem
\ref{nesscond}(b). Let us start with fixed $ j \in T,$ then
$
\Pr( j \notin \hat{T}^a) = \Pr (\hth^a _j=0) \rightarrow 0
$
by the second claim of the theorem.
Next recall that $\hth ^a$ minimizes the convex function $\Gamma ^a(\theta),$ so $0 \in \partial \Gamma ^a (\hth ^a)$ and
$
 \partial \Gamma ^a (\hth ^a)=\nabla Q(\hth ^a) + \partial \left(\lambda
 \sum_{j=1}^p \frac{ |\hth ^a _j|}{|\hbeta_j |} \right).
$
If we consider fixed $ j \notin T $ and an event $\{j \in \hat{T} ^a \},$ then
we have
\begin{equation}
\label{oracle_form2}
0=  \nabla^j Q(\hth ^a) +
\lambda
\frac{\sign (\hth_j ^a)}{|\hbeta_j |}.
\end{equation}
 From (\ref{oracle_form2}) we have
\begin{equation}
\label{oracle_form3}
\sqrt{n} \, \left|\nabla^j Q(\hth ^a) \right| =\frac{n\lambda}{ \sqrt{n}|\hbeta_j |}.
\end{equation}
The right-hand side of (\ref{oracle_form3}) tends to infinity in probability, because $n\lambda
\rightarrow \infty$ and its denominator is bounded in probability. If we can show that
the left-hand side of (\ref{oracle_form3}) is bounded in probability, then the probability of the considered event $\{j \in \hat{T} ^a \}$
tends to zero and it finishes the proof. Notice that
\begin{eqnarray}
\label{oracle_form4}
\sqrt{n}  \,  \nabla Q(\hth ^a)  &=&
- \left[ \frac{n-1}{n} \sqrt{n} A+ \bar{X}/\sqrt{n}-
\sqrt{n} \frac{X ' X}{n} \th*\right]  \\
\label{oracle_form5}
&+&
 \frac{X ' X}{n} \sqrt{n} (\hth^a- \th*).
\end{eqnarray}
The term on the right-hand side of \eqref{oracle_form4} is bounded in probability by Lemma \ref{asym}. Using LLN, \eqref{end_form} and Slutsky's theorem we can also bound \eqref{oracle_form5} in probability.

\end{proof}

\begin{proof}[Proof of Theorem \ref{asym_thresh}]
Using Lemma \ref{pomoc} with $a=1/ \sqrt{n}$ we obtain
\begin{equation}
\label{thr_form}
\sqrt{n}\left(\hth - \th* \right) \rightarrow_d - H^{-1}W ,
\end{equation}
because $\sqrt{n} \lambda \rightarrow 0.$
Fix $j \notin T,$ so $\th* _j=0.$ Then we have from \eqref{thr_form} and $\sqrt{n} \delta \rightarrow \infty$ that $\delta ^{-1}  \hth _j
\rightarrow _P 0,$ so $P(\hth ^{th}_j=0) = P( |\hth _j| < \delta) \rightarrow 1.$

Similarly, take $j \in T$ such that $\th* _j >0.$ From \eqref{thr_form} we know that $\hth _j$ is a consistent estimator of $\th* _j.$ Therefore, $P(\hth ^{th}_j >0) = P( \hth _j > \delta)$
tends to one, because $\delta \rightarrow 0.$ Argumentation for $j \in T$ such that $\th* _j <0$
is analogous. Using Theorem \ref{mult_cor} we finish the proof of the claim (a) of the theorem.

From \eqref{thr_form} we have $\sqrt{n}\left(\hth _T - \th* _T \right) \rightarrow_d - (H^{-1}W)_T.$ Moreover, we have just proved that $P(\hth ^{th}_T  = \hth _T ) \rightarrow 1.$
It finishes the proof of the claim (b).

\end{proof}

\section{Supplementary simulation results}
\label{add_exp}

In Figures \ref{FDR} and \ref{Power} we show plots of FDR and Power for estimators considered in Section \ref{numerical} of the main paper. In Figure \ref{boxplots} we have boxplots for gene expressions from Section \ref{real} of the main paper.

\begin{figure*}
\begin{center}
\includegraphics{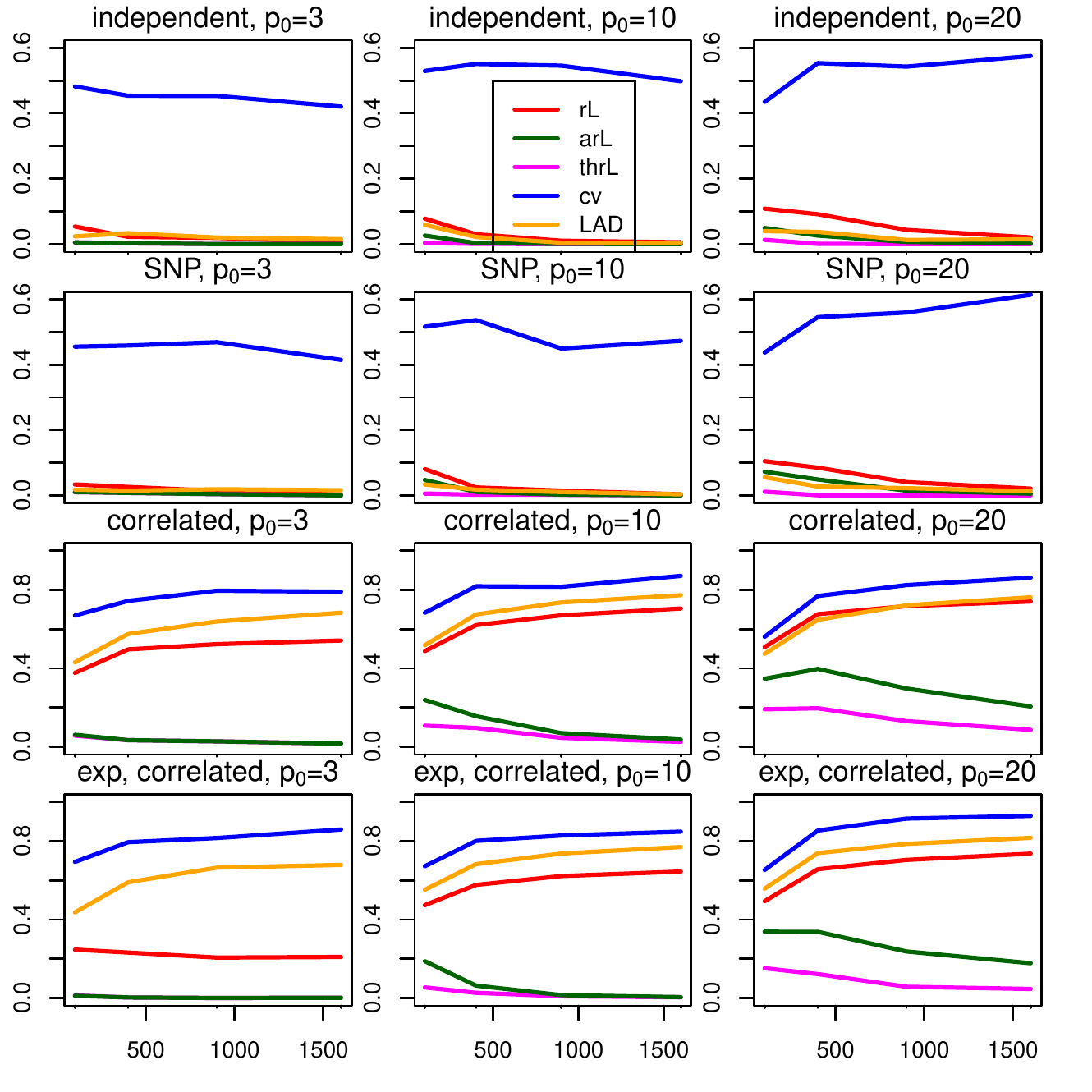}
\caption{\large{\bf Plots of FDR  for different simulation scenarios}}
\label{FDR}
\end{center}
\end{figure*}

\begin{figure*}
\begin{center}
\includegraphics{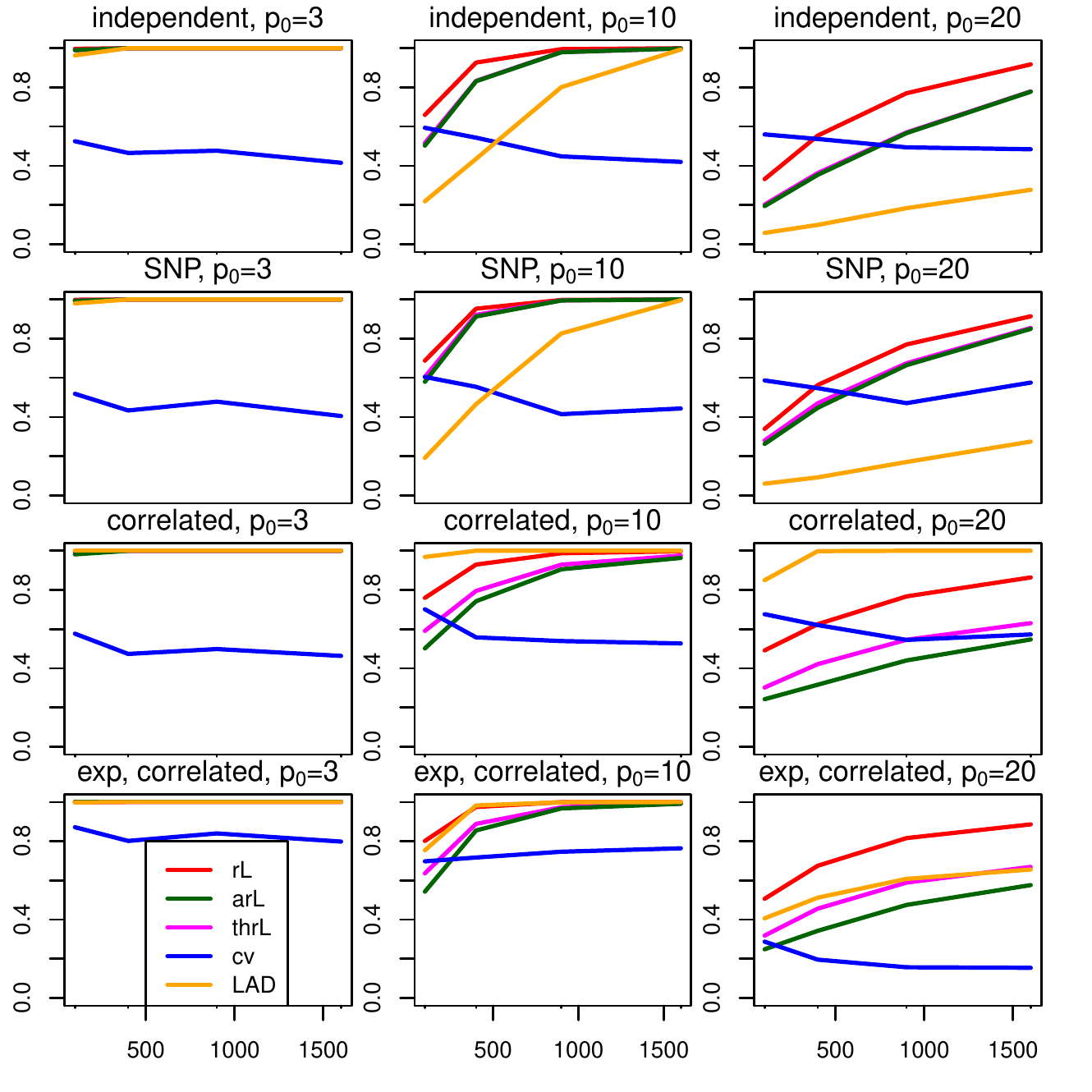}
\caption{\large{\bf Plots of Power  for different simulation scenarios}}
\label{Power}
\end{center}
\end{figure*}

\begin{figure*}
\begin{center}
\includegraphics{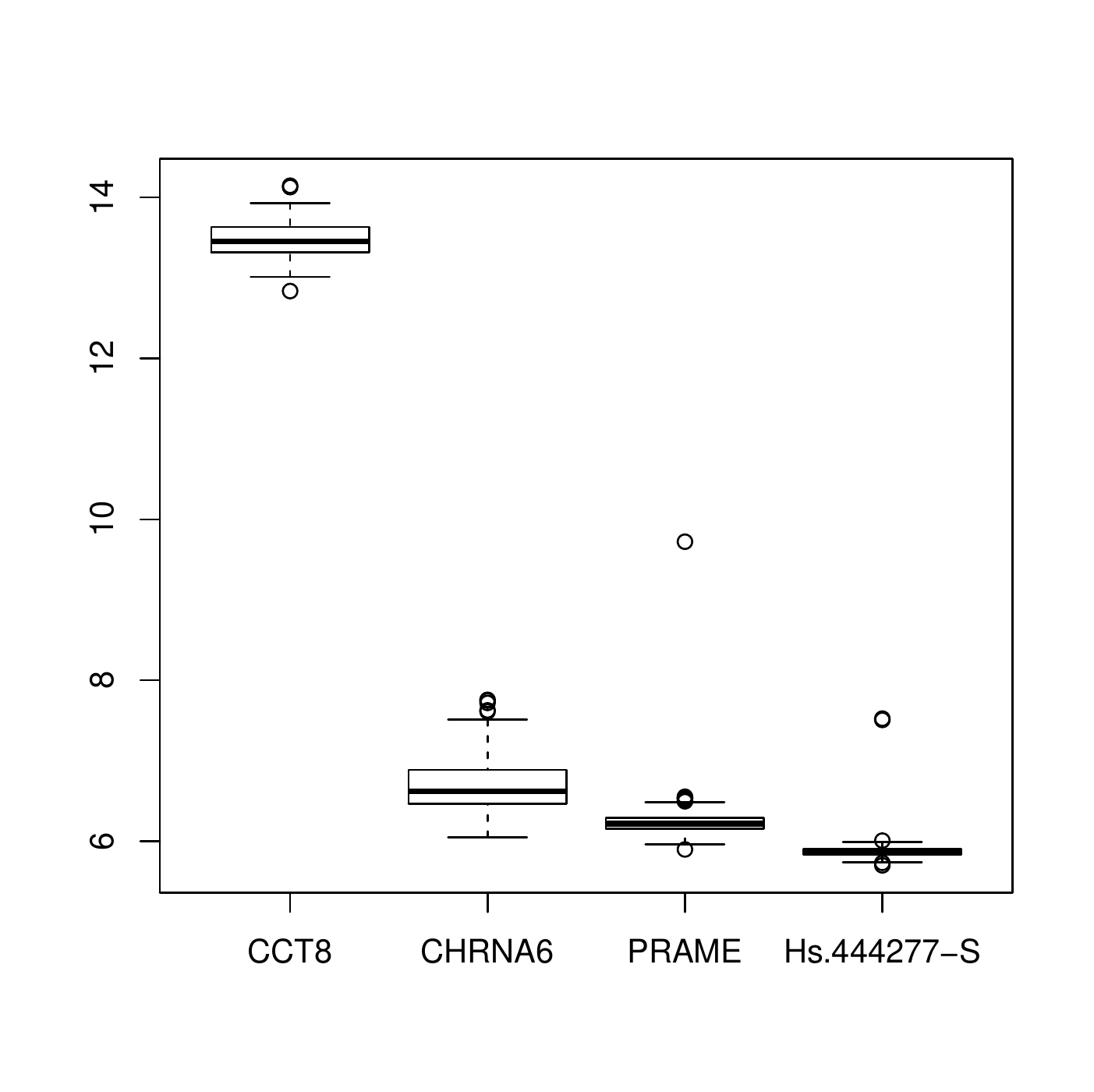}
\caption{\large{\bf Boxplots of gene expressions}}
\label{boxplots}
\end{center}
\end{figure*}

\bibliographystyle{abbrv}


\bibliography{arxiv_manuscript}

\end{document}